\documentclass[accepted]{uai2026} % after acceptance, for a revised version; 
\usepackage[american]{babel}
\usepackage{natbib} % has a nice set of citation styles and commands
\usepackage{mathtools} % amsmath with fixes and additions
\usepackage{booktabs} % commands to create good-looking tables
\usepackage{tikz} % nice language for creating drawings and diagrams
\usepackage{graphicx}
\usepackage{caption}
\usepackage{subcaption}
\usepackage{amsthm, amssymb}
\newtheorem{theorem}{Theorem}
\newtheorem{lemma}{Lemma}
\newtheorem{proposition}{Proposition}
\usepackage[most]{tcolorbox}
\usepackage{enumitem}

\newcommand{\E}{\mathbb{E}}
\newcommand{\G}{\mathcal{G}}
\newcommand\ci{\perp\!\!\!\perp}
\DeclareMathOperator{\doo}{do}
\DeclareMathOperator{\pa}{pa}
\DeclareMathOperator{\Pa}{Pa}

\newtcolorbox{procedurebox}[1][]{%
  enhanced,
  breakable,
  colback=gray!5,
  colframe=black,
  boxrule=0.8pt,
  arc=2pt,
  left=6pt,right=6pt,top=6pt,bottom=6pt,
  title={#1},
  fonttitle=\bfseries,
  coltitle=black,
  attach boxed title to top left={yshift=-1.5mm, xshift=2mm},
  boxed title style={colback=black!10, colframe=black, boxrule=0.8pt, arc=2pt, left=4pt, right=4pt, top=2pt, bottom=2pt},
}

\title{Robust Weighted Triangulation of Causal Effects Under Model Uncertainty}

\author[1,*]{Rohit Bhattacharya}
\author[1,*]{Ina Ocelli}
\author[2]{Ted Westling}
\affil[$\text{\textcolor{white}{1}}$]{%
\texttt{\hspace{0.8cm} rb17@williams.edu \hspace{1.2
cm} io1@williams.edu \hspace{0.75cm} twestling@umass.edu \vspace{0.25cm}}
}

\affil[1, *]{
    Equal contributors. Dept.~of Computer Science\\
    Williams College
  }
  
\affil[2]{%
    \texttt{} Dept.~of Mathematics and Statistics\\
    University of Massachusetts Amherst
}
\begin{document}
\maketitle

\begin{abstract}

A fundamental challenge in causal inference with observational data is correct specification of a causal model. When there is model uncertainty, analysts may seek to use estimates from multiple candidate models that rely on distinct, and possibly partially overlapping, sets of identifying assumptions to infer the causal effect, a process known as \emph{triangulation}. Principled methods for triangulation, however, remain underdeveloped. Here, we develop a framework for causal effect triangulation that combines model testability methods from causal discovery with statistical inference methods from semiparametric theory, while avoiding explicit model selection and post-selection inference problems. We propose a triangulation functional that combines identified functionals from each model with data-driven measures of model validity. We provide a bound on the distance of the functional from the true causal effect along with conditions under which this distance can be taken to zero. Finally, we derive valid statistical inference for this functional. Our framework formalizes robustness under causal pluralism without requiring agreement across models or commitment to a single specification. We demonstrate its performance through simulations and an empirical application.

\end{abstract}

\section{Introduction}
\label{sec:intro}

Causal inference with observational data is rarely conducted under a single, indisputable set of assumptions. For a single observed dataset, there may be multiple plausible causal models %--for example, an instrumental variable model \citep{angrist1996identification} and a frontdoor model \citep{pearl1995causal}---
that identify the target causal parameter under different sets of assumptions, some of which may be untestable. A common response to model uncertainty is to combine evidence from each of these models with the hope of reaching more robust conclusions than reliance on just a single one, a process commonly referred to as \emph{triangulation} \citep{thurmond2001point, farmer2006developing, lawlor2016triangulation}.

Despite its intuitive appeal, triangulated effect estimation remains underdeveloped. A prominent line of work formalizes triangulation through \emph{evidence factors}, which combine p-values across multiple analyses to test causal hypotheses \citep{rosenbaum2010evidence, rosenbaum2011some, karmakar2019integrating}. While powerful, these approaches are fundamentally geared toward hypothesis testing rather than estimation, and rely on assumptions that are often difficult to justify in observational settings. In particular, they require that different analyses do not share sources of bias and that the joint distribution of p-values satisfies certain stochastic dominance properties under the null. 
\cite{yangstatistical} leveraged the joint convergence properties of semiparametric estimators to design a robust test of the causal null hypothesis that remains valid as long as at least one model is correct. However, \cite{yangstatistical} also focused on hypothesis testing rather than effect estimation.

In this paper, we develop a general framework for triangulating causal effect estimates across multiple candidate models without requiring explicit model selection or correctness of a plurality of models, as required by voting-based triangulation procedures. Our approach instead assigns data-driven weights to each model based on testable implications of its identifying assumptions, such as conditional independence or generalized equality constraints. These weights are used to form a smooth aggregation of model-specific effect estimates, yielding an estimator that is consistent for a \emph{weighted triangulation functional}. We show that the absolute difference between this functional and the true causal parameter can be bounded as a function of (i) the maximal bias among incorrect models and (ii) our ability to separate correct and incorrect models using observed data. We use this bound to provide conditions under which the bias of the triangulation functional is small, yielding a form of robustness to misspecification of some causal models through triangulation. In particular, we show the bias can be controlled if at least one candidate model is correct and testable from observed data. 
We demonstrate the effectiveness of the proposed method through simulations and an empirical application. 

\textbf{Contributions.} Our contributions can be viewed from two complementary perspectives. First, we advance the triangulation literature by providing a principled and quantitative method for combining causal effect estimates that achieves robustness to model misspecification by leveraging testable implications in the observed data. Second, we contribute to the literature on post-selection inference, which studies valid inference after data-driven model selection. Rather than selecting a single model, our approach avoids explicit selection by using smooth weights, thereby mitigating post-selection bias while still incorporating information from model diagnostics. Unlike recent work in causal discovery \citep{gradu2025valid, chang2026post}, which focuses on learning the full causal graph, our framework targets a more modest but practically important goal: testing a minimal set of assumptions required to identify a causal parameter and combining estimators of the identified functionals.

\textbf{Related work.} The methods in \cite{rakshit2025adaptive, kang2016instrumental, sun2021multiply, yao2024deciphering} allow for some models to be misspecified, but require that a plurality of models be correctly specified. This can be a strong assumption in observational settings, where several models may fail for similar reasons. Moreover, these methods are typically tailored to specific classes of causal models, such as proxy-variable approaches or instrumental variables. Bayesian model averaging approaches \citep{horii2021bayesian, steiner2025bayesian} provide another avenue for combining estimates, but often rely on parametric assumptions such as linearity or are similarly restricted to specific model classes.

\section{Causal Graph Preliminaries}
\label{sec:prelims}

\begin{figure}[t]
	\begin{center}
		\scalebox{0.9}{
			\begin{tikzpicture}[>=stealth, node distance=1.25cm]
				\tikzstyle{square} = [draw, thick, minimum size=1.0mm, inner sep=3pt]
				\tikzstyle{format} = [thick, circle, minimum size=1.0mm, inner sep=3pt]
				
				\begin{scope}[]
					\path[->, very thick]
					node[] (z) {$Z$}
					node[right of=z] (a) {$A$} 
					node[right of=a] (m) {$M$}
					node[right of=m] (y) {$Y$}
                    node[above of=a, yshift=-0.2cm] (c) {$C$}
                    node[above of=m, yshift=-0.2cm] (u) {$U$}
					
					(a) edge[] (m)
                    (m) edge[] (y)
                    (z) edge[] (a)
                    (c) edge[] (y)
                    (c) edge[] (a)
                    (c) edge[] (m)
                    (c) edge[] (z)
                    (u) edge[red] (a)
                    (u) edge[red] (y)
					
					node[below of=a, xshift=0.5cm, yshift=0.5cm] (a) {(a) $\G(V\cup U)$} ;
				\end{scope}
				
				\begin{scope}[xshift=5cm]
					\path[->, very thick]
					node[] (z) {$Z$}
					node[right of=z] (a) {$A$} 
					node[square, right of=a] (m) {$m$}
					node[right of=m] (y) {$Y$}
                    node[above of=a, yshift=-0.2cm] (c) {$C$}
                    node[above of=m, yshift=-0.2cm] (u) {$U$}
					
					(m) edge[] (y)
                    (z) edge[] (a)
                    (c) edge[] (y)
                    (c) edge[] (a)
                    (c) edge[] (z)
                    (u) edge[red] (a)
                    (u) edge[red] (y)
					
					node[below of=a, xshift=0.5cm, yshift=0.5cm] (a) {(b) $\G(V\cup U)_{\doo(m)}$} ; 
				\end{scope}
				
			\end{tikzpicture}
		}
	\end{center}
	\vspace{-0.4cm}
	\caption{(a) A hidden variable causal DAG. (b) Graph representing intervention on $M$ from which we can read the Verma constraint $Z \ci Y \mid C$ in $P(V)/P(M\mid A, Z, C)$.}
	\label{fig:verma-graph}
\end{figure}
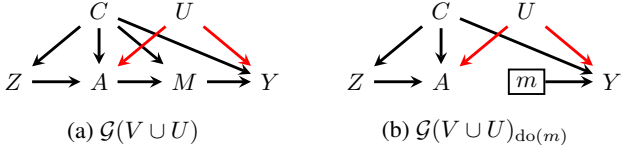

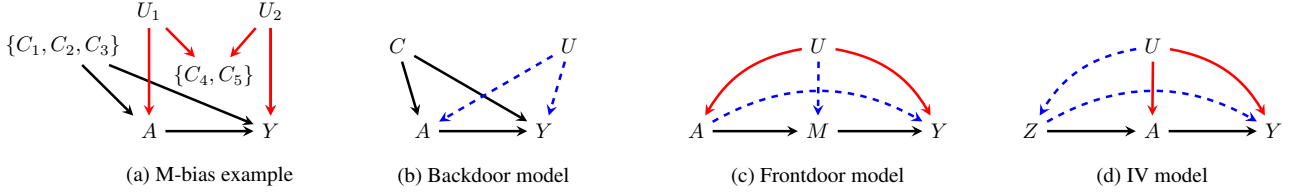
\begin{figure*}[t]
		\begin{center}
			\scalebox{0.8}{
				\begin{tikzpicture}[>=stealth, node distance=2cm]
					\tikzstyle{square} = [draw, thick, minimum size=1.0mm, inner sep=3pt]
					
					\begin{scope}[xshift=0cm]
						\path[->, very thick]
						node[] (a) {$A$} 
						node[right of=a] (y) {$Y$}
						node[above right of=a, xshift=-0.35cm, yshift=-0.5cm] (c45) {$\{C_4, C_5\}$}
                        node[above left of=a] (c123) {$\{C_1, C_2, C_3\}$}
                        node[above of=a, yshift=0cm] (u1) {$U_1$}
                        node[above of=y, yshift=0cm] (u2) {$U_2$}
                        node[below of=a, yshift=1.25cm, xshift=1cm] (label) {(a) M-bias example}
						
						(c123) edge[] (a)
						(a) edge[] (y)
						(c123) edge[] (y)
                        (u1) edge[red] (a)
                        (u1) edge[red] (c45)
                        (u2) edge[red] (y)
                        (u2) edge[red] (c45)
                        (u2) edge[red] (y)
						;
					\end{scope}

                    \begin{scope}[xshift=4.5cm]
						\path[->, very thick]
						node[] (a) {$A$} 
						node[right of=a] (y) {$Y$}
						node[above left of=a, xshift=1cm] (c) {$C$}
                            node[above right of=y, xshift=-1cm] (u) {$U$}
						node[below of=a, yshift=1.25cm, xshift=1cm] (label) {(b) Backdoor model}
						
						(c) edge[] (a)
						(a) edge[] (y)
						(c) edge[] (y)
                        (u) edge[blue, dashed] (y)
                        (u) edge[blue, dashed] (a)
						;
					\end{scope}
					
                    \begin{scope}[xshift=9cm]   
						\path[->, very thick]
						node[] (a) {$A$}
						node[right of=a] (m) {$M$}
						node[right of=m] (y) {$Y$}
                        node[above right of=m, xshift=-1.4cm] (u) {$U$} 
                        node[below of=m, yshift=1.25cm](label) {(c) Frontdoor model}
						
                        (a) edge[] (m)
                        (m) edge[] (y)
                        (u) edge[red, bend right=25] (a)
                        (u) edge[blue, dashed] (m)
                        (u) edge[red, bend left=25] (y)
                        (a) edge[blue, dashed, bend left] (y)
						;
					\end{scope}
     
					\begin{scope} [xshift=14.5cm]
						\path[->, very thick]
						node[] (z) {$Z$}
						node[right of=z] (a) {$A$}
						node[right of=a] (y) {$Y$}
                        node[above right of=a, xshift=-1.4cm] (u) {$U$} 
						node[below of=a, yshift=1.25cm] (label) {(d) IV model}
						
						(z) edge[] (a)
						(a) edge[] (y)
                        (u) edge[red] (a)
                        (u) edge[red, bend left=25] (y)
                        (u) edge[blue, dashed, bend right=25] (z)
                        (z) edge[blue, dashed, bend left] (y)
						;
					\end{scope}
				
				\end{tikzpicture}
			}
		\end{center}
        \vspace{-0.4cm}
		\caption{Causal DAGs used for motivating robust triangulation. (a) A causal DAG where uncertainty about what variables to adjust for may yield M-bias. (b, c, d) Under model uncertainty, analysts may wish to triangulate effects from each of these models that rely on qualitatively distinct assumptions---violations of these assumptions are shown via blue dashed edges.}
		\label{fig:motivation}
\end{figure*}

% Although the candidate causal models need not be graphical in nature, we frame our discussion using causal directed acyclic graphs (causal DAGs), as they provide a transparent way to represent identifying assumptions and their testable implications. Here, we first provide a brief overview of hidden variable causal DAGs and their testable implications.
Although the candidate causal models need not be graphical, we frame our discussion with causal directed acyclic graphs (causal DAGs), as they provide a transparent way to represent identifying assumptions and their testable implications.

The causal model of a DAG $\G(V)$ can be understood as distributions generated by a system of structural equations equipped with the $\doo(\cdot)$ operator \citep{pearl2009causality}. Specifically, for each variable $V_i \in V$, there is an equation of the form $V_i \gets f_i\big(\pa_\G(V_i), \epsilon_i\big),$
where $\pa_\G(V_i)$ denotes a set of values for the parents of $V_i$ in $\G$ and $\epsilon_i$ is an exogenous (latent) error term. A joint distribution $P(V)$ induced by such a system is said to be Markov with respect to $\G$, meaning that it factorizes as $P(V) = \prod_{V_i \in V} P\big(V_i \mid \Pa_\G(V_i)\big)$.  %according to the graph:
% \begin{align*}
% P(V) = \prod_{V_i \in V} P\big(V_i \mid \Pa_\G(V_i)\big).
% \end{align*}
Equivalently, $P$ satisfies the global Markov property with respect to $\G$ stated in terms of the well-known d-separation criterion: $X \ci_{\text{d-sep}} Y \mid Z \implies X \ci Y \mid Z \text{ in } P$ \citep{pearl2009causality}. The distribution $P$ is considered \emph{faithful} to $\G$ if the converse is also true, i.e., $X \ci_{\text{d-sep}} Y \mid Z \iff X \ci Y \mid Z \text{ in } P$.
%Intuitively, faithfulness rules out ``coincidental'' independences that are not due to the structure of $\G$. It is often used as an assumption in the causal discovery literature to enable learning the structure of $\G$ from data \citep{spirtes2000causation}. In a similar vein, we assume faithfulness in our work to justify the use of data-driven weights for triangulating effects based on testability of identifying assumptions in $\G$.
%Faithfulness is often assumed in the causal discovery literature to enable learning the structure of $\G$ from data \citep{spirtes2000causation}. In a similar vein, we assume faithfulness in our work to justify  data-driven weights for triangulating effects based on testability of identifying assumptions in $\G$.
% We later use faithfulness to justify the use of observed data constraints in $P$ to rule out non-identifying edges in $\G$.

In fully observed causal DAG models, counterfactual distributions arising from interventions on a set $A \subset V$, written as $P(V \setminus A \mid \doo(a))$, are identified via a truncated factorization known as the g-formula %. Here, conditional factors corresponding to each $A_i \in A$ are removed
\citep{robins1986new, pearl2009causality}:
\begin{align}
P(V \setminus A \mid \doo(a)) 
= \prod_{V_i \in V \setminus A} 
P\big(V_i \mid \Pa_\G(V_i)\big)\Big|_{A=a}.
\label{eq:g-formula}
\end{align}

% When a subset of variables, say $U$, are unobserved, identification theory is significantly more complicated---in Scenario~2 of our motivating section, we describe some popular methods used to overcome unobserved confounding and discuss why triangulating their estimates may be beneficial.
% Additionally, the constraints implied by the hidden variable DAG $\G(V \cup U)$ on the observed distribution $P(V)$ now consist of not only regular conditional independences that can be read by d-separation, but also generalized equality constraints, known as \emph{Verma constraints} \citep{robins1986new, verma1990equivalence}.  Verma constraints are obtained via d-separation in graphs corresponding to post-intervention distributions identified from the observed joint distribution.
When some variables $U$ are unobserved, identification theory becomes more complex (as in Scenario~2 of Section~\ref{sec:motivation}). In addition, the observed distribution $P(V)$ consists of not only conditional independences from d-separation, but also generalized equality constraints, known as \emph{Verma constraints} \citep{robins1986new, verma1990equivalence}. Verma constraints are obtained via d-separation in conditional DAGs corresponding to identifiable post-intervention distributions.
As an example, consider the DAG $\G(V \cup U)$ in Figure~\ref{fig:verma-graph}(a). By d-separation, $\G$ implies $Z \ci M \mid A, C$. However, there is no set  $X \subset V \setminus \{Z, Y\}$ such that $Z \ci_{\text{d-sep}} Y \mid X$. That is, there is no ordinary conditional independence between $Z$ and $Y$ in the observed joint $P(V)$. However, it is well known that there exists a Verma constraint  $Z \ci Y \mid C$ in the Markov kernel $P(V)/P(M\mid A, Z, C)$ \citep{verma1990equivalence}. %or equivalently, $\sum_a P(Z, C, a)\cdot P(Y|a, Z, C, M)$ is not a function of $Z$.
Under a causal interpretation of $\G(V\cup U)$, for any fixed value of $m$, $P(V)/P(M=m\mid A, Z, C) \vert_{M=m} = P(Z, A, C, Y \mid \doo(m))$ by the g-formula in \eqref{eq:g-formula}. %\footnote{This coincides with dropping a conditional factor of $M$ in the factorization of $\G$, similar to the g-formula in~\eqref{eq:g-formula}.}.
That is, the kernel $P(V)/P(M\mid A, Z, C)$ factorizes according to the conditional DAG  in Figure~\ref{fig:verma-graph}(b), where $M$ is now a fixed node with all incoming edges removed, and from which the Verma constraint can be read using d-separation---notice that $Z$ and $Y$ are d-separated given $C$ in $\G(V\cup U)_{\doo(m)}$.

When the measures of model validity used rely on Verma constraints, we will require $P(V)$ to be \emph{Verma constraint faithful} to $\G(V\cup U)$. That is, $X \ci Y \mid Z$ in $P(V)/P(A | B)$ implies that edges in $\G$ are such that $P(V \setminus A \mid \doo(a)) = P(V)/P(A | B)\vert_{A=a}$   and  $X \ci_{\text{d-sep}} Y \mid Z$ in  $\G_{\doo(a)}$.
%In Section~\ref{sec:applications}, we state exactly what types of faithfulness are assumed.  %Additional details on Verma constraint faithfulness can be found in \cite{shpitser2014introduction} and \cite{bhattacharya2021differentiable}. We now move on to describing two scenarios that motivate our framework.

\section{Motivating Scenarios}
\label{sec:motivation}
Suppose an analyst is interested in estimating the average causal effect  $\theta \equiv \E[Y \mid {\doo(A=1)}] - \E[Y \mid {\doo(A=0)}]$. In this section, we present two classic scenarios where proper triangulation of effect estimates from multiple plausible causal models would lead to more robust causal inference.

\textbf{Scenario 1. Adjust for all pre-treatment covariates or not?} When a set of pre-treatment covariates $L$ blocks all backdoor paths between $A$  and $Y$ we obtain a well-known instance of the g-formula, known as the backdoor formula:
\begin{align}
\!\!\theta = \sum_l P(l) \times \big( \E[Y\mid A=1, l] - \E[Y\mid A=0, l] \big).
\label{eq:backdoor}
\end{align}
A longstanding debate in causal inference is whether analysts should include \emph{all} observed pre-treatment covariates in the backdoor formula~\eqref{eq:backdoor}. Some researchers \citep{rosenbaum2002observational, ding2015adjust, rubin2009should} argue that conditional ignorability (blockage of all backdoor paths) is likely to hold only when adjusting for a large set of variables, and thus advocate adjusting for all pre-treatment covariates $C$ associated with the treatment and outcome. Other researchers emphasize that adjustment for all such variables can induce collider bias, such as M-bias \citep{pearl2009remarks, shrier2009remarks, sjolander2009propensity}. Thus, when the true  DAG is unknown---as is typical in observational studies---the analyst may face significant uncertainty about whether excluding certain covariates risks confounding bias or including them risks collider bias. This setting naturally motivates triangulation across multiple adjustment strategies.

As a concrete example, let Figure~\ref{fig:motivation}(a) be the true (but unknown) DAG. Here, adjusting for all observed pre-treatment covariates $C$  yields invalid effect estimates due to the colliders at $C_4$ and $C_5$. Due to model uncertainty, the analyst might partition the observed covariates into two groups: those that certainly do not give M-bias and are essential for adjustment, and those about which there is some uncertainty. Say they correctly deem $\{C_1, C_2, C_3\}$ to be essential, but are uncertain about $\{C_4, C_5\}$. They may then wish to triangulate estimates from multiple overlapping adjustment sets, e.g., $\{C_1, C_2, C_3\}$ that excludes all uncertain covariates, $\{C_1, C_2, C_3, C_4\}$ that includes some but not others, and $\{C_1, C_2, C_3, C_4, C_5\}$ that includes all pre-treatment covariates. In this example, only the first set leads to valid effect estimates, underscoring the importance of a triangulation procedure that is robust to incorrect identifying assumptions in multiple models as long as at least one is correct.

\textbf{Scenario 2. Use backdoor, frontdoor, or an instrument?}
In many observational settings, there is uncertainty about whether unmeasured confounding between $A$ and $Y$ is present. %If the true (unknown) causal DAG $\G(V \cup U)$ contains an unmeasured confounder of the form $A \leftarrow U \rightarrow Y$, then no subset of observed variables suffices for valid identification via the backdoor formula \eqref{eq:backdoor}.
%In such cases, standard covariate adjustment may yield biased estimates.
A natural response to this is to triangulate estimates obtained from a backdoor adjustment model (Figure~\ref{fig:motivation}(b)) with those from alternative identification strategies that explicitly allow for unmeasured $A$--$Y$ confounding but impose different structural assumptions. One such alternative is the frontdoor model (Figure~\ref{fig:motivation}(c)). The frontdoor model permits unmeasured confounding between $A$ and $Y$, but assumes (i) that $A$ affects $Y$ only through a mediator set $M$ (i.e., no direct $A \rightarrow Y$ edge), and (ii) that the unmeasured variable $U$ does not confound the $A$--$M$ or $M$--$Y$ relationships (i.e., no $U \rightarrow M$ edge). Under these assumptions, the causal effect is identified by the frontdoor formula \citep{pearl1995causal}. A conditional version that adjusts for baseline covariates $L$ is
\begin{align}
\theta = \sum_{m, l} \bigg\{ &\big(\sum_{a'}P(a', l) \cdot \E[Y\mid a', m, l]\big) \nonumber \\
&\cdot\big(P(m\mid A=1, l) - P(m\mid A= 0, l)\big) \bigg\}.
\label{eq:frontdoor}
\end{align}

Alternatively, one may use instrumental variable (IV) models (Figure~\ref{fig:motivation}(d)). An instrument $Z$ typically satisfies three structural assumptions: (i) conditional independence from the unmeasured confounder $U$ (no $U \rightarrow Z$ edge), (ii) an exclusion restriction (no $Z \rightarrow Y$ edge), and (iii) instrument relevance ($Z \rightarrow A$ exists). Under additional non-graphical assumptions, such as effect homogeneity, the causal effect is point-identified as \citep{angrist1996identification}:
\begin{align}
\theta = 
\frac{\sum_l P(l) \cdot \big(\E[Y \mid Z=1, l] - \E[Y \mid Z=0, l])}
{\sum_l P(l) \cdot \big( \E[A \mid Z=1, l] - \E[A \mid Z=0, l] \big)}.
\label{eq:iv}
\end{align}

Unlike Scenario 1, this setting involves three qualitatively distinct causal models. Each has their own unique drawbacks, but may imply testable restrictions in $P(V)$ under some assumptions of faithfulness and causal ordering of variables \citep{entner2013data, bhattacharya2022testability}. A natural approach then would be to test such implications, retain non-rejected models, and aggregate their estimates. %, e.g., via a simple average or variance-weighted average \citep{rubin1975variance}.
However, this two-step strategy introduces a post-selection inference problem, particularly when models share data, have overlapping adjustment sets, or shared sources of bias. Further, standard remedies for this, such as sample splitting \citep{hansen2000sample, newey2018cross}, can substantially reduce effective sample size when multiple models are considered. These scenarios and associated challenges motivate our new triangulation functional and inference procedure.

% \section{Data-Driven Weighted Combination Of Causal Effects}
\section{New Triangulation Method}
\label{sec:method}

We now present the general setup that we consider. Let $\theta$ denote the causal parameter of interest, such as the average or conditional average causal effect. We assume the observed data consists of $n$ IID realizations $O_1, \dots, O_n$ drawn from some unknown distribution $P$. For a finite integer $K>1$, let ${\cal M}_1 \dots, {\cal M}_K$ denote the different candidate causal models, and let $\psi_{k}$ (we suppress the dependence on $P$ for these parameters and ones that follow for notational brevity) denote the identifying functional of ${\cal M}_k$. That is, if the assumptions of ${\cal M}_k$ are true, then $\theta = \psi_k$. Let $\beta_1, \dots, \beta_K$ denote observed data parameters and ${\cal A}$ be a set of assumptions %\footnote{${\cal A}$ can also  be split up into multiple sets ${\cal A}_1, \dots, {\cal A}_k$. %, and this would not change the proposed framework.
%However, we feel it is easier for an analyst to specify a single set of joint assumptions for testing correctness of all models, rather than individual sets for each one. We elaborate on this in Section~\ref{sec:applications}.}
such that if ${\cal A}$ is true and $\beta_{k} = 0$, then the identifying assumptions of ${\cal M}_k$ are true. We do not require, however, that if ${\cal A}$ holds and  $\beta_k \not= 0$, then the assumptions of ${\cal M}_k$ must be incorrect; we allow this scenario to also correspond to the model ${\cal M}_k$ being untestable using the observed data. 

% The interpretation of each $\beta_k$ is as follows: each $\beta_k$ is an observed data parameter---for example, a regression coefficient in a parametric model, a log-odds ratio in a semiparametric model \citep{chen2007semiparametric, tchetgen2010doubly}, or a generalized covariance measure \citep{zhang2017causal, shah2020hardness, he2025on}---that can be used to test an equality constraint in $P$. These equality constraints could be regular conditional independence constraints or generalized equality constraints, such as Verma constraints \citep{robins1986new, verma1990equivalence} and tetrad constraints \citep{shafer1995generalization, silva2006learning}. The equality constraint holds in $P$ if and only if $\beta_k=0$. Thus, if $\beta_k = 0$, then under assumptions ${\cal A}$, one of which will usually be an assumption like faithfulness \citep{spirtes2000causation},   one or more specific edges in a causal graph are absent. Furthermore, absence of these specific edges is a necessary condition of the identifying assumptions of ${\cal M}_k$. Examples of such tests proposed in the causal discovery literature include independence constraints in \cite{entner2013data} for testing the backdoor model, Verma constraints in \cite{bhattacharya2022testability} for testing the front-door model, and tetrad constraints in \cite{xie2024automating} for testing the proximal causal inference model.

The role of each $\beta_k$ is to encode a testable implication of the identifying assumptions of model ${\cal M}_k$. Examples of such testable implications from the causal discovery literature include conditional independence constraints for backdoor models \citep{entner2013data}, Verma constraints for frontdoor models \citep{bhattacharya2022testability}, and tetrad constraints for proxy-based models \citep{xie2024automating}. That is, each $\beta_k$ is an observed data parameter constructed so that it is zero if and only if a particular equality constraint, such as the aforementioned ones, holds in the observed distribution $P$. More specifically, $\beta_k$ may be a regression coefficient in a parametric model, a log-odds ratio in a semiparametric model \citep{chen2007semiparametric}, or a generalized covariance measure \citep{shah2020hardness, he2025on, bergen2026the}. Under assumptions ${\cal A}$ (typically including a faithfulness condition and partial knowledge of temporal ordering of variables), $\beta_k = 0$ implies the absence of specific edges in $\G$, and  absence of these edges is sufficient for the identifying assumptions of ${\cal M}_k$ to hold. Thus, whether $\beta_k = 0$ provides empirical evidence about the validity of ${\cal M}_k$.

\subsection{Triangulation functional}
The above setup motivates the following naive triangulation functional: $\psi_\text{naive} = \frac{\sum_{k=1}^{K} \mathbb{I}(\beta_k = 0) \cdot \psi_{k}}{\sum_{j=1}^{K} \mathbb{I}(\beta_j = 0)}$,
% This setup motivates a naive triangulation functional:
% \begin{align}
% \psi_\text{naive} = \frac{\sum_{k=1}^{K} \mathbb{I}(\beta_k = 0) \cdot \psi_{k}}{\sum_{j=1}^{K} \mathbb{I}(\beta_j = 0)},
% \label{eq:naive_triangulation}
% \end{align}
where $\mathbb{I}(\cdot)$ is the indicator function. That is, $\psi_{\text{naive}}$ is an average over all $\psi_k$ derived from models ${\cal M}_k$ for which $\beta_k = 0$. %If assumptions ${\cal A}$ used to justify tests of model correctness hold, then $\psi_{\text{naive}}$ has the following \emph{causal robustness} property: $\psi_{\text{naive}} = \theta$ if  the identifying assumptions of at least one model ${\cal M}_k$ are correct and testable from the observed data distribution $P$. This is because  $\psi_{\text{naive}}$ can then be interpreted as averaging effect estimates from models where $\theta=\psi_k$.
If assumptions ${\cal A}$ used to test model correctness hold, $\psi_{\text{naive}}$ is \emph{causally robust}: $\psi_{\text{naive}}=\theta$ if at least one model is correct and testable, as it reduces to an average of only those $\psi_k$ such that $\psi_k=\theta$.

% \begin{figure}[t]
%     \centering
%     \includegraphics[scale=0.4]{figures/dirac_delta.png} % 
%     \vspace{-0.3cm}
%     \caption{Gaussian kernel approximation of $\mathbb{I}(\beta_k = 0)$.}
%     \label{fig:dirac-delta}
% \end{figure}

We consider this to be naive as filtering based on $\mathbb{I}(\beta_k = 0)$ is unlikely to succeed for any model ${\cal M}_k$ when $\beta_k$ is unknown and must be estimated from finite samples. %\footnote{In some settings, $\beta_k=0$ because the assumptions of ${\cal M}_k$ hold by design; our method can be easily adapted to such settings.}%\footnote{It may be possible in some settings to argue that $\beta_k = 0$ if the assumptions of ${\cal M}_k$ hold by design. Our method naturally allows for such settings by using the known $\beta_k$ rather than estimating it.} 
Instead, we propose  a smooth approximation of the  naive functional by replacing $\mathbb{I}(\beta_k = 0)$ with Gaussian kernels of the form $\delta_a(\beta_k) = \frac{1}{|a|\sqrt{\pi}} e^{-\left(\beta_k / a\right)^2}$,
% \begin{align}
% \delta_a(\beta_k) = \frac{1}{|a|\sqrt{\pi}} e^{-\left(\beta_k / a\right)^2},
% \label{eq:dirac-delta}
% \end{align}
where $a > 0$ is a constant that controls the sharpness of the approximation. As \( a \to 0 \), the function becomes  concentrated around \( \beta_k = 0 \). %, mimicking the behavior of an indicator function
%; see Figure~\ref{fig:dirac-delta}.
Our proposed triangulation functional is then given by
\begin{align}
{\psi} = \sum_{k=1}^{K} w_k \psi_{k}, \quad \text{where} \quad w_k = \frac{\delta_a(\beta_k)}{\sum_{j=1}^K \delta_a(\beta_j)}.
\label{eq:triangulation_rule}
\end{align}
This functional exhibits an \emph{approximate} causal robustness property due to smoothing via the kernels, as stated below.
\begin{theorem}
    Suppose assumptions ${\cal A}$ used to justify the tests of model correctness hold, and let $\mathcal{C} \subseteq \{1, \dotsc, K\}$ and $\mathcal{I} = \{1, \dotsc, K\} \setminus \mathcal{C}$ be the subsets of indices $k$ such that model $\mathcal{M}_k$ is correct and incorrect, respectively. Then
    \begin{equation}
        |\psi - \theta| \leq   \frac{ \max_k |\psi_k - \theta|}{1 + D_a}
        \label{eq:bound}
    \end{equation}
    where $D_a = \left(\sum_{k \in \cal{C}} \delta_a(\beta_k)\right) / \left(\sum_{k \in \cal{I}} \delta_a(\beta_k)\right)$. Further, if at least one model ${\cal M}_k$ is correct and testable using the observed data, then $D_a \geq e^{\varepsilon^2/ a^2} / |\cal{I}|$, where $\varepsilon = \min_{k \in \cal{I}} |\beta_k|$.
    \label{thm:robustness}
\end{theorem}
\vspace{-1em}
The proof is in Appendix~\ref{app:theorem-proof}. Theorem~\ref{thm:robustness} demonstrates that the absolute difference between the triangulation functional $\psi$ and the true causal parameter of interest $\theta$ is bounded by the maximal bias of the functionals $\psi_k$ from the incorrect causal models divided by  a ``discrimination factor" $1 + D_a$. The discrimination factor depends both on the weighting function chosen and the true values of the testing functionals $\beta_1, \dotsc, \beta_K$. Roughly speaking, when the $\beta_k$'s from correct causal models are small in magnitude compared to those from incorrect causal models on the scale determined by $\delta_a$, then the discrimination factor is large, and the absolute difference between $\psi$ and $\theta$ is reduced. We note that the second statement of Theorem~\ref{thm:robustness} implies that if at least one model ${\cal M}_k$ is both correct and testable using the observed data and all testing functionals $\beta_k$ in incorrect models are non-zero, then $\lim_{a \to 0} D_a = \infty$, so that $\lim_{a \to 0} \psi = \theta$.

To make the robustness property more concrete, consider setting the Gaussian kernel parameter to $a=0.1$ and suppose there is one correct model $\mathcal{C}=\{1\}$ and two incorrect models $\mathcal{I}=\{2,3\}$. Let $\varepsilon=\min_{k\in\mathcal{I}}|\beta_k|=0.2$, indicating that violations of the incorrect models are only weakly detectable since the value of $\beta_k$ is close to $0$ even when ${\cal M}_k$ is incorrect. From Theorem~\ref{thm:robustness}, this yields $1+D_a \geq 1+e^{(0.2/0.1)^2}/2 \approx 30$, implying that the bias of the triangulation functional is at most the maximal bias among incorrect models divided by $30$.

While smaller values of $a$ reduce the absolute difference between the triangulated functional $\psi$ and the target causal effect $\theta$, in practice we cannot set $a$ too small because we only have estimates of $\psi_k$ and $\beta_k$. Thus, we must balance the variability of these estimators against the bias induced by $a > 0$. We return to this discussion and provide concrete recommendations for setting $a$ in Section~\ref{subsec:kernel}.

We also note that if \emph{no} candidate model is both correct and testable, then the triangulation functional does not converge to $\theta$ as $a\to 0$. %ill-defined.
Fortunately, this behavior is straightforward to diagnose in practice: the estimated values of $\psi$ often diverge  due to $\sum_{j=1}^K \delta_a(\beta_j) \approx 0$. The normalizing term, however, also poses challenges in finite samples even when a correct and testable model ${\cal M}_k$ exists in theory, e.g., due to sampling variability. We briefly address this issue of numerical stability in the following subsection on inference.

\subsection{Inference Procedure}
\label{subsec:estimation}

We now discuss our approach to estimation and inference for the triangulation functional $\psi$.  Suppose that for each $k$, we have estimators $\psi_{k, n}$ and $\beta_{k, n}$ of each identifying functional $\psi_k$ and each measure of model correctness $\beta_k$, respectively. %That is, $\text{plim}_{n\to\infty}\psi_{k, n} = \psi_k$ and $\text{plim}_{n\to\infty}\beta_{k, n} = \beta_k$.
Our triangulation estimator is then
{
\begin{align}
\!\!{\psi_n} = \sum_{k=1}^{K} w_{k, n} \psi_{k, n}, \  \text{for } \ w_{k, n} = \frac{\delta_a(\beta_{k, n})}{\lambda_n + \sum_j \delta_a(\beta_{j, n})},
\label{eq:triangulation_estimator}
\end{align}}%
where $\lambda_n > 0$ is an additional term used to ensure numerical stability of the weights $w_{k, n}$ even when the normalizing function in the denominator is close to zero. To avoid affecting first-order asymptotics of the estimator, we require $\lambda_n = o(n^{-1/2})$. A simple choice is to fix $\lambda_n = 1/n$. Note that $\lambda_n$ in the triangulation functional is optional. In practice, we find in our simulations and data application that we do not need it and that using the well-known log-sum-exp trick \citep{blanchard2021accurately}  to compute the normalizing function is sufficient to prevent numerical instability. However, we keep it in our exposition for the sake of completeness.

Since $\psi_{1, n}, \dots, \psi_{K, n}, \beta_{1, n}, \dots, \beta_{K, n}$ are constructed using a common observed dataset, they are typically dependent, even asymptotically. We suppose that under a set of statistical regularity conditions ${\cal S}$ (e.g., rates of convergence and complexity constraints on the nuisance estimators), we can construct asymptotically linear estimators of both $\beta_k$ and  $\psi_{k}$ with influence functions $\phi_{\beta_k}$ and $\phi_{\psi_k}$, respectively, for each $k$. That is, under ${\cal S}$ we have $\beta_{k, n} - \beta_k  =  \frac{1}{n} \sum_{i=1}^{n}\phi_{\beta_k}(O_i) + o_p(n^{-1/2})$ and $\psi_{k, n} - \psi_k  =  \frac{1}{n} \sum_{i=1}^{n} \phi_{\psi_k}(O_i) + o_p(n^{-1/2})$.
% \begin{align*}
% &\beta_{k, n} - \beta_k  =  \frac{1}{n} \sum_{i=1}^{n}\phi_{\beta_k}(O_i) + o_p(n^{-1/2}), \\
% &\psi_{k, n} - \psi_k  =  \frac{1}{n} \sum_{i=1}^{n} \phi_{\psi_k}(O_i) + o_p(n^{-1/2}),
% \end{align*}
Here $\phi_{\beta_k}$ and $\phi_{\psi_k}$ are assumed to satisfy  $\E[\phi_{\beta_k}] = \E[\phi_{\psi_k}] = 0$,  $\E[\phi^2_{\beta_k}] < \infty$, and $\E[\phi^2_{\psi_k}] < \infty$. These estimators could be parametric or semiparametric in nature; we will discuss specific approaches to estimation more below. %We will discuss approaches to constructing asymptotically linear estimators below.

Asymptotic linearity of any finite collection of estimators implies  joint convergence to a multivariate normal distribution \citep{van2000asymptotic}. Let $\kappa = [\beta_1, \dots, \beta_K, \psi_1, \dots, \psi_K]$ and $\kappa_n$ denote the corresponding vector of estimates. Then, $n^{1/2} (\kappa_n - \kappa) \xrightarrow{d} N(0, \Sigma)$, where $\Sigma$ is a $2K \times 2K$ covariance matrix of the influence functions of each $\psi_k$ and $\beta_k$. Each entry of $\Sigma$ is of the form $\E[\phi_{\beta_j}, \phi_{\beta_k}], \E[\phi_{\psi_j}, \phi_{\psi_k}]$, or $\E[\phi_{\beta_j}, \phi_{\psi_k}]$.
% {\footnotesize
% \begin{align*}
% \begin{bmatrix}
%     \E(\phi_{\beta_1}^2) & \E(\phi_{\beta_1}, \phi_{\beta_2}) & \dots & \E(\phi_{\beta_1}, \phi_{\psi_K}) \\
%     \E(\phi_{\beta_2}, \phi_{\beta_1}) & \E(\phi_{\beta_2}^2) & \dots & \E(\phi_{\beta_2}, \phi_{\psi_K}) \\
%     \vdots & \vdots & \ddots & \vdots \\
%     \E(\phi_{\psi_K}, \phi_{\beta_1}) & \E(\phi_{\psi_K}, \phi_{\beta_2}) &\dots & \E(\phi_{\psi_K}^2)
% \end{bmatrix}.
% \end{align*}}%
% \begin{align*}
% n^{1/2} (\kappa_n - \kappa) \xrightarrow{d} N(0, \Sigma),
% \end{align*}
Applying the delta method then gives the asymptotic distribution of the triangulation estimator as,
\begin{align}
n^{1/2}(\psi_n - \psi) \xrightarrow{d} N(0, \gamma^T\Sigma\gamma),
\label{eq:convergence}
\end{align}
where $\gamma$ is a vector of partial derivatives defined as,
\begin{align}
\gamma = \bigg[\ \frac{\partial \psi}{\partial \beta_1}, \dots,  \frac{\partial \psi}{\partial \beta_K}, \frac{\partial \psi}{\partial \psi_1}, \dots,  \frac{\partial \psi}{\partial \psi_K} \ \bigg].
\label{eq:gamma}
\end{align}
The following lemma also gives us closed form expressions for computing the vector $\gamma$. The proof is in Appendix~\ref{app:partials}.
\begin{lemma}
\label{lem:partials}
The partial derivatives of $\psi$ are  %$\frac{\partial {\psi}}{\partial \psi_k} = w_k$ and
{
\begin{align*}
\frac{\partial {\psi}}{\partial \psi_k} = w_k \quad \text{ and } \quad
\frac{\partial \psi}{\partial \beta_k} = %\frac{2\psi_k\beta_k w_k }{a^2} \left( \sum_{j\not=k} w_j  - \frac{\lambda_n + \sum_{j\not=k}\delta(\beta_j)  }{ \lambda_n + \sum_{j=1}^K \delta(\beta_j) }\right).
\frac{2\beta_k w_k}{a^2} \left(\psi - \psi_k \right).
\end{align*}}
\end{lemma}
% The proof of Lemma~\ref{lem:partials} can be found in Appendix~\ref{app:partials}.

We now suggest three  approaches to constructing asymptotically linear estimators and associated approaches to valid inference.  Our preferred  strategy is influence function-based estimation. In this approach, the %efficient
influence functions $\phi_{\beta_k}$ and $\phi_{\psi_k}$ are derived explicitly and used to construct $\beta_{k,n}$ and $\psi_{k,n}$ using, e.g., the one-step construction \citep{bickel1982adaptive} or estimating equations \citep{chernozhukov2018double}. %Alternatively, influence functions can be computed numerically.
A benefit to this approach is that machine learning estimators can be used for nuisance functions, and asymptotic linearity follows under sufficient rates of convergence of these nuisance estimators. %(in addition to complexity constraints, which can be avoided using cross-fitting).

When the influence functions are available in closed form---several of which have been derived in the context of causal graphs \cite{jung2021estimating, bhattacharya2022semiparametric, guo2024average}---we can  construct consistent estimators $\phi_{\beta_k,n}$ and $\phi_{\psi_k, n}$ of the influence functions. A consistent estimator $\Sigma_n$ of  $\Sigma$ is then simply the sample covariance matrix of the influence function estimators. %For example, we estimate $\E[\phi_{\beta_j}\phi_{\beta_k}]$ with $\frac{1}{n}\sum_{i=1}^n \phi_{\beta_k,n}(O_i) \phi_{\beta_j, n}(O_i)$.
An estimator $\gamma_n$ of $\gamma$ can be obtained by plugging in $\psi_{k,n}$ and $\beta_{k,n}$ into the form of $\gamma$ provided in Lemma~\ref{lem:partials}. A variance estimator of $\psi_n$ is then given by $\widehat{\text{var}}(\psi_n) = \gamma_n^T \Sigma_n \gamma_n / n$, which can be used to construct Wald-type confidence intervals for $\psi$. %We also note that if influence function estimators are not available a partial bootstrap can be used instead. %The final estimate of the variance is then obtained by computing empirical estimates of $\gamma_n$ and $\Sigma_n$ from the data. That is,
% \begin{align}
%     \widehat{\text{var}}(\psi_n) = \gamma_n^T \Sigma_n \gamma_n.
%     \label{eq:variance_estimator}
% \end{align}

For some functionals $\psi_k$ or $\beta_k$, explicit influence functions might not be available in the current literature. In this case, one option is to use plug-in estimators $\psi_{k,n}$ and $\beta_{k,n}$. However, plug-in estimators are often only asymptotically linear when correctly-specified parametric models are employed for nuisance estimators. In this case, the empirical bootstrap yields asymptotically valid inference under mild smoothness conditions on the parametric models \citep{efron1994introduction}. Specifically, on each bootstrap sample, the parametric nuisance estimators are re-estimated and plugged in to obtain bootstrap estimates $\psi_{k,n}^*$ and $\beta_{k,n}^*$. These estimates are then substituted as in~\eqref{eq:triangulation_estimator} to obtain a bootstrap estimate $\psi_n^*$. The distribution of bootstrap estimates can then be used to construct a confidence interval in the standard ways.

Finally, if plug-in estimators with nuisance estimators based on machine learning are used, the resulting estimators may not be asymptotically linear due to excess bias, and hence the empirical  bootstrap is not guaranteed to work. In this case, %the $m$-out-of-$n$ bootstrap, also known as 
\emph{subsampling} may still be asymptotically valid \citep{politis2001asymptotic}. Subsampling involves drawing $b$ random subsamples of size $m < n$ %\textcolor{red}{with replacement}
from the full dataset, computing estimates $\psi_m^{(1)}, \dots, \psi_m^{(b)}$ for each subsample, and constructing a $(1-\alpha)$ confidence interval based on the $\alpha/2$ and $1-\alpha/2$ quantiles of these estimates. A sufficient condition for this procedure to yield valid inference is that $\tau_n(\psi_n-\psi)$ converges to a non-degenerate distribution, which can be true even when the estimators are not asymptotically linear.

%Thus, as long as $\gamma^T\Sigma\gamma > 0$, we obtain a sufficient condition for valid inference via subsampling---convergence of $n^{1/2}(\psi_n-\psi)$ to a non-degenerate normal distribution.

%implies the following marginal convergence result for each $\psi_{k, n}$ and $\beta_{k, n}$ \citep{van2000asymptotic},
%\begin{align*}
%&n^{1/2} (\beta_{k, n} - \beta_k) \xrightarrow{d} N(0, \sigma^2_{\beta_k}), \text{ and } \\
%&n^{1/2} (\psi_{k, n} - \psi_k) \xrightarrow{d} N(0, \sigma^2_{\psi_k}).
%\end{align*}%
%In addition, asymptotic linearity 

%Since the IF based estimators are also asymptotically linear, the convergence result in~\eqref{eq:convergence} still holds. However, now the covariance matrix 

% In summary (see also the box below), for inference of $\psi$ in \eqref{eq:triangulation_rule}, influence function-based estimators are preferred, followed by plug-in estimators with parametric nuisance models, as subsampling intervals tend to be overly conservative due to discarding substantial data (typically $m = n^{4/5}$).
We summarize our recommendations below. The order in the if-elif-else logic reflects our preferred inference strategies.

\begin{procedurebox}[Inference Procedure for Triangulation]

\begin{enumerate}[leftmargin=*, itemsep=4pt]

\item Construct estimators $\psi_{k,n}$ and $\beta_{k,n}$ for each model ${\cal M}_k$ and obtain the point estimate $\psi_n$ using \eqref{eq:triangulation_estimator}.

% \item Obtain the point estimate $\psi_n$ using \eqref{eq:triangulation_estimator}.

\item If every $\psi_{k,n}$ and $\beta_{k,n}$ is influence-function-based, %proceed to Step 3; otherwise, proceed to Step 4.
let $\widehat{\mathrm{SE}}(\psi_n) =
\sqrt{\gamma_n^\top \Sigma_n \gamma_n/n}$ and 
construct a $(1-\alpha)$ CI as $\psi_n \pm z_{1-\alpha/2}\,\widehat{\mathrm{SE}}(\psi_n)$.
%\item Else, if $\psi_{k, n}$ and $\beta_{k, n}$ are plug-in estimators based on parametric nuisance models: Construct a $(1-\alpha)$ CI as $[q_{\alpha/2},\, q_{1-\alpha/2}]$, the empirical quantiles of estimates $\{\psi_n^{(1)},\dots,\psi_n^{(b)}\}$ obtained via empirical bootstrap and re-fitting of nuisance models.
\item Else, if $\psi_{k, n}$ and $\beta_{k, n}$ are plug-in estimators based on parametric nuisance models: Construct a $(1-\alpha)$ CI as $[2\psi_n - q_{1-\alpha/2},\, 2\psi_n - q_{\alpha/2}]$, where $q_{p}$ is the $p^{\text{th}}$ empirical quantile of estimates $\{\psi_n^{(1)},\dots,\psi_n^{(b)}\}$ obtained via empirical bootstrap and re-fitting of nuisance models.

\item Else, construct a $(1-\alpha)$ CI as $[2\psi_n - q_{1-\alpha/2},\, 2\psi_n - q_{\alpha/2}]$, based on the empirical quantiles of estimates $\{\psi_m^{(1)},\dots,\psi_m^{(b)}\}$ obtained via subsampling with $m=n^{4/5}$ and re-fitting of nuisance models.

\end{enumerate}

\end{procedurebox}

\subsection{Setting the kernel bandwidth}
\label{subsec:kernel}
Finally, we address setting the  kernel bandwidth parameter $a$. Our recommendations in this section apply when estimators that converge at the rate $n^{-1/2}$, which includes influence function-based estimators and plug-in estimators based on parametric nuisance models (options 2 and 3 of the display in Section~\ref{subsec:estimation}). Setting $a$ when using estimators with slowers rates of convergence, such as many plug-in estimators based on flexible nuisance models, is an interesting problem left to future work.

We recommend choosing $a$ such that (1) $n^{1/2} a \to \infty$ and (2) $a \sqrt{\log(n)} \to 0$. Examples of this include setting $a=n^{-1/3}$ and $a=1/\log(n)$. In our numerical experiments and data application we choose $a=n^{-1/3}$, and also test this setting against other kernel bandwidths that lie outside of our recommended range. Our recommendation is based on the following theoretical analysis.%; we also test our recommendation empirically in Section~\ref{sec:applications}.}

From Theorem~\ref{thm:robustness} and  Section~\ref{subsec:estimation}, the bias of the triangulation estimator $\psi_n$ for the true parameter $\theta$ decays at a rate proportional to $e^{-1/a^2} + o_P(n^{-1/2})$, while its standard deviation (SD) goes to zero at rate $n^{-1/2}$ when using influence function-based or parametric estimators. If $a \sqrt{\log(n)} \to 0$ (i.e., $a$ goes to zero faster than $1/\sqrt{\log(n)}$), then $e^{-1/a^2} = o(1/\sqrt{n})$, so that the bias goes to zero faster than the SD. Hence, as long as at least one model is correct and testable and regularity conditions for nuisance estimation hold, then our methods produce asymptotically valid inference for the true causal effect if  $a \sqrt{\log(n)} \to 0$.

Now let ${\cal C}'$ denote the set of indices of all models ${\cal M}_k$ that are both correct and testable. When the estimators $\beta_{k, n}$ are asymptotically linear, then $\sqrt{n}\beta_{k, n}$ is asymptotically normal if model $k$ is correct and testable, and $c_n \beta_{k, n} \to_p \infty$ for any sequence $c_n \to \infty$ otherwise. This implies the following: (i) If $a$ goes to zero faster than than $n^{-1/2}$, then asymptotically $\psi_n$ is equal to $\psi_{k, n}$ for some $k$ randomly selected from ${\cal C}'$; (ii) If $a$ goes to zero at the rate $n^{-1/2}$, then asymptotically $\psi_n$ is a random weighted combination of all $\psi_{k, n}$ for $k \in {\cal C}'$, with the joint distribution of weights depending on the asymptotic covariance of $\sqrt{n}\beta_{k, n}$ for $k \in {\cal C}'$; (iii) If $a$ goes to zero slower than $n^{-1/2}$, then asymptotically $\psi_n$ is an equally-weighted average of all $\psi_{k, n}$ for $k \in {\cal C}'$.

In our view, scenario (iii) is the most desirable behavior. Asymptotically, scenario (i) will result in focusing exclusively on the model whose testing functional estimator is closest to zero in any particular sample, which is overly narrow. The asymptotic covariance of the testing functionals is not necessarily related to the asymptotic variance of $\psi_{k, n}$, so there is not a good reason to weight estimators as in scenario (ii). Thus, (iii) arguably reflects the most reasonable behavior. It follows that we recommend choosing $a$ such that $n^{1/2} a \to \infty$ (i.e., $a$ goes to zero slower than $n^{-1/2}$).

% \textcolor{purple}{Further, in practice we recommend taking $a$ such that $n^{1/2} a \to \infty$ (i.e., $a$ goes to zero slower than $n^{-1/2}$) due to the following. Let $S$ be the set of correct and testable models. If the $\hat\beta_k$ estimators are asymptotically linear, then $n^{1/2}\hat\beta_k$ is asymptotically normal if model $k$ is correct and testable, and $c_n \hat\beta_k \to_p \infty$ for any sequence $c_n \to \infty$ otherwise. Hence:}

% Hence, our final recommendation is to choose $a$ satisfying 1) $n^{1/2} a \to \infty$ and 2) $a \sqrt{\log(n)} \to 0$.

\section{Applying The Triangulation Procedure In Practice}
\label{sec:applications}

\begin{figure}[t]
	\begin{center}
		\scalebox{0.9}{
			\begin{tikzpicture}[>=stealth, node distance=1.25cm]
				\tikzstyle{square} = [draw, thick, minimum size=1.0mm, inner sep=3pt]
				\tikzstyle{format} = [thick, circle, minimum size=1.0mm, inner sep=3pt]
				
				\begin{scope}[xshift=0cm]
                \path[->, very thick]
                node[] (a) {$A$}
                node[left of=a] (z) {$Z$}
                node[right of=a, xshift=0.3cm] (y) {$Y$}
                node[above right of=a, xshift=-0.2cm, yshift=-0.3cm] (c45) {$C_{45}$}
                node[above left of=a] (c123) {$C_{123}$}
                node[above of=a, yshift=0.3cm] (u1) {$U_1$}
                node[above of=y, yshift=0.3cm] (u2) {$U_2$}
                % node[below of=a, yshift=0.8cm, xshift=1cm] (label) {(a) M-bias example}
                
                (c123) edge[] (a)
                (a) edge[] (y)
                (c123) edge[] (y)
                (u1) edge[red] (a)
                (u1) edge[red] (c45)
                (u2) edge[red] (y)
                (u2) edge[red] (c45)
                (u2) edge[red] (y)
                (z) edge[] (a)
                (u1) edge[red, dashed, bend right=55] (z)
                ;
                \end{scope}
				
				\begin{scope}[xshift=4cm]
                \path[->, very thick]
                node[] (a) {$A$}
                node[left of=a] (z) {$Z$}
                node[right of=a] (m) {$M$}
                node[right of=m] (y) {$Y$}
                node[above of=a] (c) {$C$}
                node[above of=y] (u) {$U$}
                (a) edge[] (m)
                (m) edge[] (y)
                (z) edge[] (a)
                (c) edge[] (y)
                (c) edge[] (a)
                (c) edge[] (m)
                (u) edge[red] (a)
                (u) edge[red] (y)
                (z) edge[blue, dashed, bend left=35] (m)
                ;
            \end{scope}
				
			\end{tikzpicture}
		}
	\end{center}
	\vspace{-0.4cm}
	\caption{Causal DAG used in our simulations to demonstrate robustness to (a) M-bias in some adjustment sets, and (b) misspecification of backdoor, frontdoor, or IV models.}%---when the blue dashed edge is present in (b), both backdoor and IV are incorrect, and only frontdoor remains correct.}
	\label{fig:app-dags}
\end{figure}
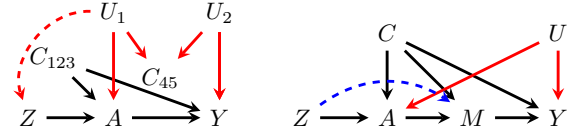

We now illustrate practical applications of our general method, along with numerical experiments and an empirical data application. The first subsection outlines a general framework for triangulating causal effect estimates derived from backdoor models with different candidate adjustment sets. The second subsection describes triangulation with backdoor, frontdoor model, and IV models. %In both cases we also provide empirical evidence of the effectiveness of triangulation using simulated data based on the examples in Section~\ref{sec:motivation};
Explicit descriptions of the data generating processes and estimators are in Appendix~\ref{app:dgps} and~\ref{app:estimators} respectively, with proofs in Appendix~\ref{app:proposition-proofs}.

\subsection{Multiple Candidate Adjustment Sets}
\label{subsec:multi-bdoor}

We first establish a set of assumptions ${\cal A}$ under which we can use observed data parameters $\beta_k$ to test the validity of each proposed adjustment set corresponding to different causal models ${\cal M}_k$. Let $C$ denote all pre-treatment covariates under consideration for adjustment and $Z$ denote an ``anchor'' variable, such that the following assumptions hold:
\begin{align}
&{\cal A}_1: P \text{ is faithful wrt a causal DAG } \G(V\cup U) \nonumber \\
&{\cal A}_2: \G \text{ satisfies the causal ordering } \{Z, C\} < A< Y\nonumber \\
&{\cal A}_3: Z \rightarrow A \rightarrow Y \text{ exists in } \G
\label{eq:assumptions-backdoor}
\end{align}
Note that ${\cal A}_1$ above is ordinary faithfulness as we do not rely on Verma constraints for this particular application, ${\cal A}_2$ simply imposes some weak background knowledge on the causal ordering of the variables, and ${\cal A}_3$ is a relevance assumption that ensures the testable implications are not trivial and actually rule out non-identifying edges in $\G$.\footnote{Since $A\rightarrow Y$ exists, the sharp causal null is assumed false. For robust causal null hypothesis tests, see \cite{yangstatistical}.} %Versions of these assumptions appear in \cite{entner2013data, gultchin2020differentiable, shah2022finding}.

\begin{proposition}
  \label{prop:backdoor-tests}
  Let ${\cal M}_k$ be a backdoor model that proposes $W \subseteq C$ as an adjustment set and $\beta_k$ be any observed data parameter such that $\beta_k = 0 \iff Y \ci Z \mid A, W$. Under assumptions ${\cal A}$ in \eqref{eq:assumptions-backdoor}, $\beta_k = 0$ implies ${\cal M}_k$ is correct, i.e., $\psi_k = \theta$, where $\psi_k$ is the backdoor formula \eqref{eq:backdoor} with $L=W$.
\end{proposition}

    A natural choice for $\beta_k$ above is the log-odds ratio. For distinct sets of variables $A, B, C$ and a choice of reference values $a_0, b_0$, the odds ratio function $\text{OR}(A, B \mid C)$ is a non-parametric measure of association given by \citep{chen2007semiparametric},
    \begin{align*}
    &\text{OR}(a, b \mid c) = \frac{P(a\mid b, c)}{P(a_0\mid b, c)}\times \frac{P(a_0\mid b_0, c)}{P(a\mid b_0, c)},
    \end{align*}
and $\log(\text{OR}(A, B \mid C)) = 0 \iff A \ci B \mid C$. Though the log-odds ratio is a function, it is most commonly treated as a single scalar parameter in a semiparametric model; see for e.g., \citep{chen2007semiparametric, tchetgen2010doubly, malinsky2019potential}. We adopt this setup, computing log-odds ratios of the form $\log(\text{OR}(Y, Z \mid A, W))$  for different choices of adjustment sets $W$ as given by ${\cal M}_1, \dots, {\cal M}_k$, and use these as our $\beta_k$ in the triangulation functional \eqref{eq:triangulation_rule}.

We return to the M-bias scenario in Section~\ref{sec:motivation} as a concrete application. The candidate adjustment sets are ${\cal M}_1: \{C_1, C_2, C_3\}$, ${\cal M}_2: \{C_1, C_2, C_3, C_4\}$, and ${\cal M}_3: \{C_1, C_2, C_3, C_4, C_5\}$, where only ${\cal M}_1$ is correct. Figure~\ref{fig:app-dags}(a) shows a DAG with an anchor variable $Z$ satisfying assumptions ${\cal A}$ in \eqref{eq:assumptions-backdoor} with or without the red dashed edge. Note that choices for anchor variables in the causal discovery literature are often candidate IVs, such as $Z$ in Figure~\ref{fig:app-dags}(a). While one could include IV-based estimates in the triangulation functional, when a valid adjustment set exists, backdoor estimators are more efficient and avoid homogeneity assumptions.

By Proposition~\ref{prop:backdoor-tests}, $\beta_1 = \log(\text{OR}(Y, Z \mid A, C_{123}))=0$, while $\beta_2$ and $\beta_3$ are non-zero. We evaluate two settings, one where $\varepsilon = \min_{k \in {\cal I}} = |\beta_k| = 0.71$, and another where $\varepsilon=0.36$. We achieve the latter by excluding $U_1 \rightarrow Z$ in Figure~\ref{fig:app-dags}(a).  Per Theorem~\ref{thm:robustness}, the second setting is more challenging for our triangulation estimator, as the absolute bias $|\psi - \theta|$ is higher, since it is harder to detect the incorrect models. Here we use influence function-based estimators of the triangulation functional. As shown in Figure~\ref{fig:results} (top row), the estimator remains causally robust in both cases, despite the plurality of models being incorrect. We require more samples, however, to obtain reliable estimates when $\varepsilon = 0.36$, as predicted by our theory. To test our variance computation proposal and Wald-type CI construction, we also compute the coverage of the estimator for both $\psi$ and the target parameter $\theta$. We report coverage for both, as these parameters are different in general, albeit with a small bounded difference. At $n=5000$, the estimator achieves the nominal coverage of $95\%$ for both $\psi$ and $\theta$ in both settings.

\subsection{Backdoor, frontdoor, and IV}
\label{subsec:bdoor-fdoor-iv}

\begin{figure*}
    \begin{center}
        \includegraphics[width=0.9\linewidth]{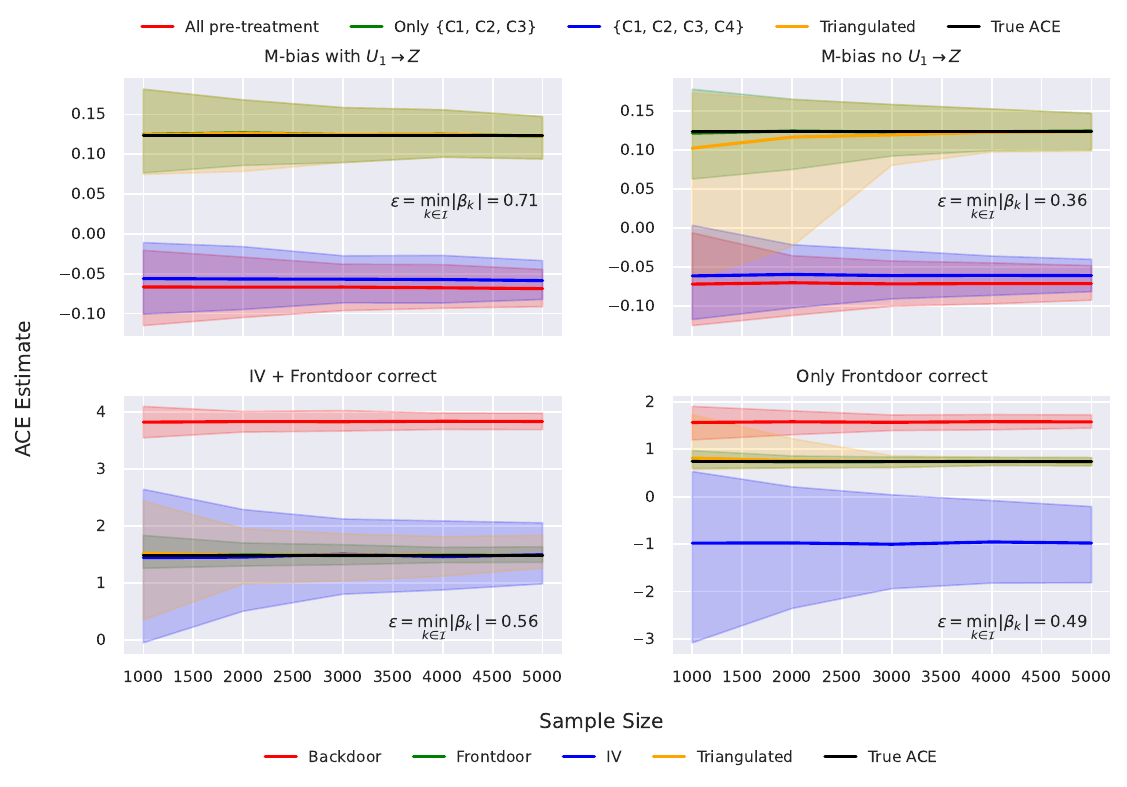}
    \end{center}
    \vspace{-0.5cm}
    \caption{Point estimates averaged over $200$ trials; shaded bands correspond to $2.5$ and $97.5$ percentiles of the estimates.}
    \label{fig:results}
\end{figure*}

Similar to the previous subsection, we first establish  assumptions under which we can use observed data to test each ${\cal M}_k$, where ${\cal M}_k$ could be a backdoor, frontdoor, or IV model. As before, let $C$ denote pre-treatment covariates used for adjustment and $Z$ an anchor variable. As stated earlier, the anchor $Z$ is often a candidate IV, and we will treat it as such in this subsection. Further, let $M$ be a mediator set such that:
\begin{align}
&{\cal A}_1: P \text{ is Verma faithful wrt a causal DAG } \G(V\cup U) \nonumber \\
&{\cal A}_2: \G \text{ satisfies the ordering } \{Z, C\} < A< M < Y\nonumber \\
&{\cal A}_3: Z \rightarrow A \rightarrow M \rightarrow Y \text{ exist in } \G 
\label{eq:assumptions-backdoor-frontdoor-iv} \\
&{\cal A}_4: \text{If } U_i \in U \text{ causes } M \text{ it also causes } Y \nonumber
\end{align}
Note that ${\cal A}_1$ includes ordinary faithfulness as a special case when the intervention set is empty, ${\cal A}_2$  imposes a causal ordering of the variables as before, and ${\cal A}_3, {\cal A}_4$ are relevance assumptions that ensure the testable implications are not trivial and actually rule out non-identifying edges in $\G$. %Versions of these assumptions are used for designing tests for the frontdoor and IV models in \cite{bhattacharya2022testability} and the frontdoor model in \cite{shah2023front}.

\begin{proposition}
\label{prop:bdoor-fdoor-iv}
Let ${\cal M}_1, {\cal M}_2, {\cal M}_3$ be backdoor, frontdoor, and IV models respectively and $\beta_1, \beta_2, \beta_3$ be observed data parameters that are zero iff the independence $Y \ci Z \mid A, C$, the Verma constraint $Y \ci Z \mid C \text{ in } p(V)/p(M | A, Z, C)$, and the independence $M \ci Z \mid A, C$ hold in $P$ respectively.
% \begin{align*}
%     &\beta_1 = 0 \iff Y \ci Z \mid A, C \\
%     &\beta_2 = 0 \iff Y \ci Z \mid C \text{ in } p(V)/p(M\mid A, Z, C) \\
%     &\beta_3 = 0 \iff M \ci Z \mid A, C.
% \end{align*}
Under assumptions ${\cal A}$ in \eqref{eq:assumptions-backdoor-frontdoor-iv}, 
% if $\beta_1 = 0$ implies  ${\cal M}_1$ is correct with adjustment set $C$, if $\beta_2 = 0$ implies the frontdoor model ${\cal M}_2$ is correct with mediator set $M$ and adjustment set $C \cup \{Z\}$, and if both $\beta_2=\beta_3=0$ implies the IV model ${\cal M}_3$ is correct.
\begin{align*}
&\beta_1 = 0 \implies \theta = \psi_1 \text{ in \eqref{eq:backdoor} with } L=C, \\
&\beta_2 = 0 \implies \theta = \psi_2 \text{ in \eqref{eq:frontdoor} with } M=M, L=C \cup\{Z\}, \\
&\beta_2=\beta_3 = 0 \implies \theta = \psi_3 \text{ in \eqref{eq:iv} with } Z=Z, L = C.
\end{align*}
\end{proposition}

We will use $\log(\text{OR}(Y, Z \mid A, C))$ and $\log(\text{OR}(M, Z \mid A, C))$ as $\beta_1$ and $\beta_3$ respectively. For $\beta_2$, we use a reweighted log-odds  $\log(\text{OR}_{\widetilde{P}}(Y, Z \mid C))$ defined with respect to the distribution $\widetilde{P} = p(V)/P(M\mid A, Z, C)$. This is estimated using a reweighted regression procedure, similar to those described in \cite{robins1997estimation, bhattacharya2022testability, robins2000marginal}. 

%Different from the previous subsection, 
The validity of the IV model also relies on assumptions of homogeneity of the treatment effect---we assume such conditions hold apriori and focus only on structural assumptions of the causal models. Further, per Proposition~\ref{prop:bdoor-fdoor-iv}, testability of the IV model relies on testability of the frontdoor model, since $\beta_2$ must be zero as well---other known tests of the IV assumptions are also in over-identified models \citep{kitagawa2015test}. Thus, to test the IV model we use $\widetilde{\beta_3} = \beta_2 + \beta_3$  under the assumption that these parameters do not exactly cancel out, which can be taken to be a form of faithfulness.

We return to Scenario 2 in Section~\ref{sec:motivation}. Figure~\ref{fig:app-dags}(b) is an example of a DAG that satisfies assumptions ${\cal A}$ in \eqref{eq:assumptions-backdoor-frontdoor-iv}. Since influence function-based estimators of parameters encoding Verma constraints, such as $\log(\text{OR}_{\widetilde{P}}(Y, Z \mid C)))$, are underdeveloped, we rely on plug-in estimators for each $\psi_k, \beta_k$ with parametric nuisance estimators. This also serves as a test of the bootstrapping branch of our triangulation procedure.  We test two scenarios, one in which frontdoor and IV are correct and testable while backdoor is incorrect, and the other in which only frontdoor  is correct and testable (using the blue dashed edge in Figure~\ref{fig:app-dags}(b)). Figure~\ref{fig:results} shows that our estimator maintains robustness in both  scenarios. When both frontdoor and IV are correct and testable, our triangulation estimator also exhibits lower variance than relying solely on the IV model. Thus, the estimator has benefits even when an analyst may be certain of some identifying assumptions if these do not yield the most efficient estimator. At $n=5000$ and when both frontdoor and IV are correct, we achieve coverage of $97\%$ and $96\%$ for $\psi$ and $\theta$ respectively; when only the frontdoor model is correct, the coverage is just below nominal coverage---$93\%$ for both $\psi$ and $\theta$.

We also use Scenario 2 to  test our recommendations in Section~\ref{subsec:kernel} for setting $a$. In particular, we use the scenario where the frontdoor and IV models are correct, and try four different settings of the bandwidth parameter: $a=n^{-1/3}$, $1/\log(n)$, $1/\sqrt{\log(n)}$, and $1/n$. The results are shown in Figure~\ref{fig:different-as}. The choices $a=n^{-1/3}$ and $a=1/\log(n)$ lie within our recommended range and have roughly similar performance. Setting $a=1/\sqrt{\log(n)}$ takes $a$ to zero too slowly, resulting in excess bias and invalid inference for the true parameter $\theta$. Finally, setting $a=1/n$ takes $a$ to zero faster than our recommendation and results in larger variance.

\begin{figure}[t]
    \includegraphics[scale=0.54]{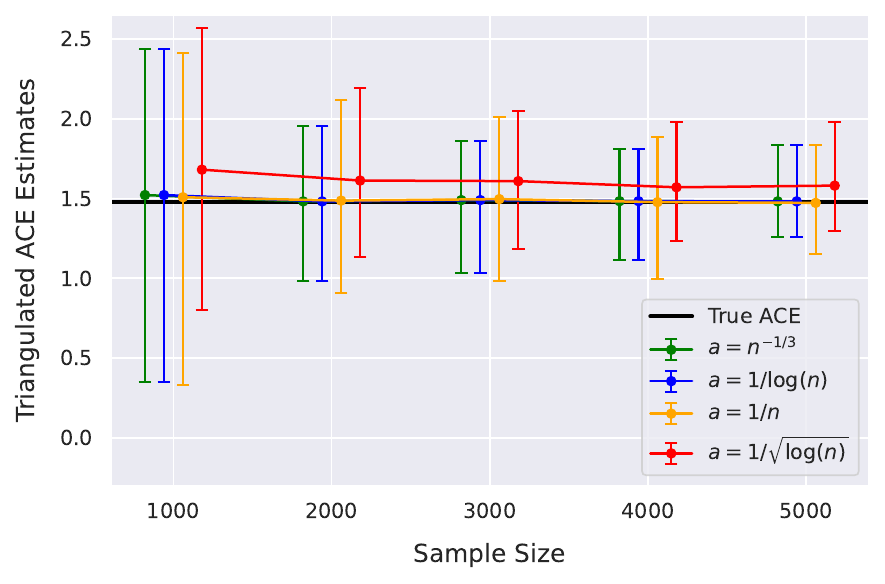} % 
    \vspace{-0.3cm}
    \caption{Triangulated point estimates using different settings of $a$ over 200 trials where frontdoor and IV are correct.}
    \label{fig:different-as}
\end{figure}

\subsection{Framingham Data Application}
\label{subsec:framingham}

\renewcommand{\arraystretch}{1.2}
{\small
\begin{table}[ht]
\centering
\begin{tabular}{|c|c|c|}
\toprule
\textbf{Method} & ${\widehat{\mathrm{{\bf ACE}}}}$ & ${ \beta_{k, n}, {w_{k, n}}}$ \\
\midrule
Backdoor & $0.086\ (0.048, 0.115)$ & $-0.04, 0.45$ \\
Frontdoor & $0.011\  (0.006, 0.016)$ & $-0.02, 0.55$ \\
IV & $-3.12\  (-14.7, 11.09)$ & $-0.41, \approx0$ \\
Triangulation  & $0.044 (0.026, 0.062)$ & -- \\
\bottomrule
\end{tabular}
\caption{Results for the empirical data application.}
\label{tab:triangulation}
\end{table}}

%\textbf{Framingham Data Application.}
We use data from the Framingham Heart Study \citep{kannel1968framingham} to estimate the effect of blood glucose levels on coronary heart disease. We treat blood glucose levels as a binary treatment variable, where $A=1$ corresponds to higher than average levels. The outcome $Y$ is a binary indicator of coronary heart disease. Similar to our setup in Section~\ref{subsec:bdoor-fdoor-iv}, we consider backdoor, frontdoor, and IV models, where  $Z=\text{educational attainment}$, $C=\text{sex}$, and $M=\text{hypertension}$. We intentionally adjust for only one confounder  $C$, so that the assumptions of backdoor in particular are difficult to justify, and so triangulation with other estimates from frontdoor or IV may be important.

% Our results are in Table~\ref{tab:triangulation}. The last column shows the weights provided to each model in the triangulation functional. The frontdoor model has the highest weight, while the IV model has the lowest. However, the backdoor model also has fairly high weight---we hypothesize that this is because sex is one of the strongest confounding variables for heart disease and high blood sugar, so that the assumptions of the backdoor model may still approximately hold. The triangulation functional accordingly provides an effect estimate that lies inbetween the backdoor and frontdoor estimates, but mostly disregards the IV estimates based on model diagnostics. Overall, the triangulated estimate suggests a 4.7\% increase in coronary heart disease in the population when intervening to set blood glucose levels to be higher than average. We believe this strikes a good balance through triangulation of the possible overestimate from the backdoor model---as a check, we also confirm that including other confounders, such as age, in the adjustment set does indeed increase the weight that would be provided to the backdoor model in the triangulation functional, while also reducing its effect estimate to about 0.05.
Results are shown in Table~\ref{tab:triangulation}, where the last column reports model weights. The frontdoor model receives the highest weight and the IV model the lowest. The backdoor model also retains substantial weight---likely because sex is one of the strongest confounders for both heart disease and blood glucose, making its assumptions approximately valid. The triangulated estimate lies between the backdoor and frontdoor estimates, largely discounting the IV estimate based on diagnostics, and indicates a 4.4\% increase in coronary heart disease under intervention to raise blood glucose. As a check, we confirm including additional confounders (e.g., age) does indeed increase the weight applied to the backdoor model, while  its point estimate drops to about 0.05.

\section{Discussion and Conclusion}
In this work we proposed a general framework for triangulating causal effects that uses data-driven weights based on measures of model validity. We showed that our triangulation functional  trades robustness to causal model misspecification in exchange for some modest bias, as well as additional assumptions like faithfulness and background knowledge. That is, while robustness through triangulation is desirable, our framework is not without limitations.

Faithfulness, in particular, is an assumption that is often contested in the causal discovery literature. While violations of faithfulness are rare in a measure-theoretic sense \citep{meek1995strong, boeken2024bayesian}, near violations of faithfulness can be fairly frequent \citep{uhler2013geometry}. In our framework, near violations of faithfulness can result in small $\varepsilon$ in Theorem~\ref{thm:robustness} and thus large bias. Analysis for cases where $\varepsilon$ shrinks with $n$ would provide deeper  understanding of the impact of near violations of faithfulness. This is an area of potential future research.

We also developed inference strategies with frequentist guarantees for our proposed triangulation functional, and demonstrated their performance through numerical studies and a data application. This included challenging scenarios in which the plurality of models vote similarly on an incorrect causal effect. Extensions to this may focus on sensitivity analysis, deriving other practical scenarios in which the framework can be used, and data-driven selection of the kernel parameter $a$. Another interesting avenue for future research is to incorporate ideas from Bayesian paradigms for model averaging into our framework, particularly those which provide robustness to unfaithfulness and the need for point identification \citep{silva2016causal}.

% \begin{contributions} % will be removed in pdf for initial submission 
% \end{contributions}

\begin{acknowledgements} 

The authors would like to thank Daniel Malinsky for helpful discussions on influence function-based estimation of odds ratios,  Hyunseung Kang and Oliver Dukes for helpful discussions on subsampling, and five anonymous
reviewers for their helpful peer review. RB would like to thank the Isaac Newton Institute for Mathematical Sciences for the support and hospitality during the programme \emph{Causal inference: From theory to practice and back again} when some of the work on this paper was undertaken. This work was supported by: EPSRC grant EP/Z000580/1 (RB),  NSF CRII grant 2348287 (RB), and NSF DMS 2113171 (TW).

\end{acknowledgements}

% \clearpage

% References
\bibliography{references}

\newpage

\onecolumn

\title{Robust Weighted Triangulation of Causal Effects Under Model Uncertainty\\(Supplementary Material)}
\maketitle

In this supplement we provide proofs that were omitted from the main paper for space, as well as details of data generating processes and the specific estimators used in our numerical experiments and data application.

\appendix

\section{Proof Of Theorem~\ref{thm:robustness}}
\label{app:theorem-proof}
\begin{proof}
   Since $\sum_k w_k = 1$, $w_k \geq 0$, and $\psi_k = \theta$ for any $k \in \cal{C}$, we have 
\begin{align*}
    |\psi - \theta|&= \left| \sum_{k=1}^{K} w_k \psi_{k} - \theta \right| \\
    &= \frac{\left| \sum_{k=1}^{K} \delta_a(\beta_k) \left[ \psi_{k} - \theta \right] \right|}{\sum_{k=1}^K \delta_a(\beta_k)} \\
    &= \frac{\left| \sum_{k \in \cal{I}} \delta_a(\beta_k) \left[ \psi_{k} - \theta \right] \right|}{\sum_{k \in \cal{I}} \delta_a(\beta_k) + \sum_{k \in \cal{C}} \delta_a(\beta_k)}  \\
    &\leq \frac{ \sum_{k \in \cal{I}} \delta_a(\beta_k) \max_k \left|\psi_{k} - \theta \right|}{\sum_{k \in \cal{I}} \delta_a(\beta_k) + \sum_{k \in \cal{C}} \delta_a(\beta_k)}  \\
    &= \frac{ \max_k \left|\psi_{k} - \theta \right|}{1 + D_a}.
\end{align*}
 Now by the definition of $\delta_a$,
 \begin{align*}
     D_a &= \frac{\sum_{k \in \cal{C}} \delta_a(\beta_k)}{\sum_{k \in \cal{I}} \delta_a(\beta_k)} = \frac{\sum_{k \in \cal{C}} e^{-\beta_k^2 / a^2}}{\sum_{k \in \cal{I}} e^{-\beta_k^2 / a^2}}.
 \end{align*}
Under our assumptions, $\beta_k = 0$ for any model that is correct and testable. If there is at least one correct and testable model, then  $e^{-\beta_k^2 / a^2} = 1$ for this model, and so $\sum_{k \in \cal{C}} e^{-\beta_k^2 / a^2} \geq 1$. For the denominator, by monotonicity of $x \mapsto e^{-(x/a)^2}$ for $x \geq 0$, we get $\sum_{k \in \cal{I}} e^{-\beta_k^2 / a^2} \leq |\mathcal{I}| e^{-\varepsilon^2 / a^2}$. The result follows.
\end{proof}

\section{Deriving The Partial Derivative Vector In Lemma~\ref{lem:partials}}
\label{app:partials}

To obtain the variance of the combined estimator using the delta method, we require the following vector of partial derivatives,
\begin{align}
\gamma = \bigg[\ \frac{\partial \psi}{\partial \beta_1}, \dots,  \frac{\partial \psi}{\partial\beta_K}, \frac{\partial \psi}{\partial \psi_1}, \dots,  \frac{\partial \psi}{\partial \psi_K} \ \bigg].
\end{align}
We divide the task of deriving the partial derivatives $\frac{\partial \psi}{\partial \beta_k}$ and $\frac{\partial \psi}{\partial \psi_k}$ for all $k \in \{1, \dots, K\}$ into two subsections as follows.

\subsection{Partial derivative of $\psi$ with respect to $\beta_k$}

First note that by plugging in the definition of $\psi$ and by linearity of differentiation we have,
\begin{align}
\frac{\partial {\psi}}{\partial \beta_k} = \frac{\partial}{\partial \beta_k} \left( \sum_{i=1}^{K} w_i \psi_{i} \right) = \sum_{i=1}^{K} \left( \frac{\partial (w_i \psi_i) }{\partial \beta_k} \right).
\label{eq:partial-betai-1}
\end{align}
For each term in \eqref{eq:partial-betai-1}, we can apply the product rule to get,
\begin{align}
\frac{\partial {\psi}}{\partial \beta_k} = \sum_{i=1}^K \left (  w_i \frac{\partial \psi_i }{\partial \beta_k} + \psi_i \frac{\partial w_i }{\partial \beta_k} \right).
\end{align}
The estimators $\psi_i$ are not functions of $\beta_k$ (even when $i=k$). Thus, $\frac{\partial \psi_i}{\partial \beta_k}$ is always $0$. So the derivative simplifies to, 
\begin{align}
\frac{\partial {\psi}}{\partial \beta_k} = \sum_{i=1}^K  \psi_i \frac{\partial w_i }{\partial \beta_k}.
\label{eq:partial-betai-1}
\end{align}
Plugging in the definition for weights $w_i$ we get,
\begin{align}
\frac{\partial {\psi}}{\partial \beta_k} = \sum_{i=1}^K \psi_i \times \frac{\partial}{\partial \beta_k} \left( \delta(\beta_i) / \big(\lambda_n + \sum_{j=1}^K \delta(\beta_j) \big) \right).
\label{eq:partial-betai-2}
\end{align}
We now separate the above terms in the summation based on whether $i=k$ or not, as these are the only two cases that lead to substantively different partial derivates with respect to $\beta_k$:
\begin{align}
\frac{\partial {\psi}}{\partial \beta_k} = \psi_k \times \frac{\partial}{\partial \beta_k} \left( \delta(\beta_k) / \big(\lambda_n + \sum_{j=1}^K \delta(\beta_j) \big) \right) + \sum_{i\not=k} \psi_i \times \frac{\partial}{\partial \beta_k} \left( \delta(\beta_i) / \big(\lambda_n + \sum_{j=1}^K \delta(\beta_j) \big) \right).
\label{eq:partial-betai-3}
\end{align}
We first deal with,
\begin{align*}
\frac{\partial}{\partial \beta_k} \left( \delta(\beta_k) / \big(\lambda_n + \sum_{j=1}^K \delta(\beta_j) \big) \right).
\end{align*}
By the quotient rule,
\begin{align}
\frac{\partial}{\partial \beta_k} \left( \delta(\beta_k) / \big(\lambda_n + \sum_{j=1}^K \delta(\beta_j) \big) \right) = \frac{\frac{\partial \delta(\beta_k)}{\partial \beta_k} \big( \lambda_n + \sum_{j=1}^K \delta(\beta_j) \big) - \delta(\beta_k) \frac{\partial \sum_{j=1}^K \delta(\beta_j) }{\partial \beta_k }}{ \left( \lambda_n + \sum_{j=1}^K \delta(\beta_j) \right)^2 }.
\label{eq:quotient-1}
\end{align}
We introduce two helper derivatives of the Dirac delta function to simplify \eqref{eq:quotient-1}:
\begin{equation}
\frac{\partial\  \delta(\beta_j)}{\partial \beta_k} = \frac{\partial}{\partial \beta_k } \left( \frac{1}{|a|\sqrt{\pi}} e^{-\left(\frac{\beta_j}{a}\right)^2}  \right)= \begin{cases}
-\frac{2\beta_k}{a^2} \delta (\beta_k)\ \hfill \text{ when } \ j = k\\
0 \hfill \text{ when }\ j\not=k.
\end{cases}
\label{eq:helper-derivs}
\end{equation}
Plugging these into \eqref{eq:quotient-1} and simplifying gives,
\begin{align}
    \frac{\partial}{\partial \beta_k} \left( \delta(\beta_k) / \big(\lambda_n + \sum_{j=1}^K \delta(\beta_j) \big) \right) &= \frac{\left( -\frac{2\beta_k}{a^2} \delta (\beta_k)\right) \big( \lambda_n + \sum_{j=1}^K \delta(\beta_j) \big) - \delta(\beta_k) \left( -\frac{2\beta_k}{a^2} \delta (\beta_k) \right) }{ \left( \lambda_n + \sum_{j=1}^K \delta(\beta_j) \right)^2 } \\
    &= \frac{ \left( -\frac{2\beta_k}{a^2} \delta (\beta_k)\right)\left(\lambda_n + \sum_{j\not=k}\delta(\beta_j) \right) }{\left( \lambda_n + \sum_{j=1}^K \delta(\beta_j) \right)^2 } \\
    &=  -\frac{2\beta_k  w_k }{a^2}  \frac{\left(\lambda_n + \sum_{j\not=k}\delta(\beta_j) \right) }{\left( \lambda_n + \sum_{j=1}^K \delta(\beta_j) \right) } \label{eq:beta-i-term}.
\end{align}
In the above, the first equality comes from plugging in the helper derivatives, the second follows from factorizing a common term, and the third comes from merging terms that correspond to $w_k$.
Now we tackle the second set of terms in \eqref{eq:partial-betai-3} involving,
\begin{align*}
\frac{\partial}{\partial \beta_k} \left( \delta(\beta_i) / \big( \lambda_n + \sum_{j=1}^K \delta(\beta_j) \big) \right),
\end{align*}
where $i\not=k$. By the quotient rule again,
\begin{align}
\frac{\partial}{\partial \beta_k} \left( \delta(\beta_i) / \big( \lambda_n + \sum_{j=1}^K \delta(\beta_j) \big) \right) = \frac{\frac{\partial \delta(\beta_i)}{\partial \beta_k} \big( \lambda_n + \sum_{j=1}^K \delta(\beta_j) \big) - \delta(\beta_i) \frac{\partial \sum_{j=1}^K \delta(\beta_j) }{\partial \beta_k }}{ \left( \lambda_n + \sum_{j=1}^K \delta(\beta_j) \right)^2 }.
\label{eq:quotient-2}
\end{align}
Plugging in the helper derivatives in \eqref{eq:helper-derivs} and simplifying gives,
\begin{align}
\frac{\partial}{\partial \beta_k} \left( \delta(\beta_i) / \big( \lambda_n + \sum_{j=1}^K \delta(\beta_j) \big) \right) &= \frac{0 - \delta(\beta_i) \left( -\frac{2\beta_k}{a^2} \delta (\beta_k) \right) }{ \left( \lambda_n + \sum_{j=1}^K \delta(\beta_j) \right)^2 } \\
&= \frac{2\beta_k}{a^2} \frac{\delta(\beta_k) \delta(\beta_i)}{\left(\lambda_n + \sum_{j=1}^K \delta(\beta_j) \right)^2} \\
&= \frac{2\beta_k w_k w_i}{a^2}. \label{eq:beta-k-term}
\end{align}
Plugging \eqref{eq:beta-i-term} and \eqref{eq:beta-k-term} back into \eqref{eq:partial-betai-3} gives us the following expression for the partial derivative of the combined estimator $\psi$ with respect to $\beta_k$:
\begin{align}
\frac{\partial \psi }{\partial \beta_k} = -\psi_k  \left (\frac{2\beta_k  w_k }{a^2}  \frac{\left(\lambda_n + \sum_{j\not=k}\delta(\beta_j) \right) }{\left( \lambda_n + \sum_{j=1}^K \delta(\beta_j) \right) } \right) + \sum_{i\not=k} \psi_i \left( \frac{2\beta_k w_k w_i}{a^2} \right).
\label{eq:partial-betai-almost-done}
\end{align}
This can be further simplified to yield the final expression as,
\begin{align}
\frac{\partial \psi }{\partial \beta_k} &= \frac{2\beta_k w_k}{a^2} \left(\sum_{i\not=k} w_i\psi_i -\psi_k   \frac{\left(\lambda_n + \sum_{j\not=k}\delta(\beta_j) \right) }{\left( \lambda_n + \sum_{j=1}^K \delta(\beta_j) \right) } + w_k\psi_k - w_k\psi_k \right) \\
&= \frac{2\beta_k w_k}{a^2} \left(\psi - \psi_k \left[ \frac{\left(\lambda_n + \sum_{j\not=k}\delta(\beta_j) \right) }{\left( \lambda_n + \sum_{j=1}^K \delta(\beta_j) \right) } + w_k\right] \right) \\
&= \frac{2\beta_k w_k}{a^2} \left(\psi - \psi_k \right).
\label{eq:partial-betai-simplify}
\end{align}

\subsection{Partial derivative of $\psi$ with respect to $\psi_k$}

Here we again apply the linearity of differentiation to get,
\begin{align}
\frac{\partial \psi }{\partial \psi_k} = \frac{\partial}{\partial \psi_k} \left( \sum_{i=1}^{K} w_i \psi_{i} \right) = \sum_{i=1}^{K} \left( \frac{\partial (w_i \psi_i) }{\partial \psi_k} \right).
\end{align}
Applying the product rule gives us,
\begin{align}
    \frac{\partial \psi }{\partial \psi_k} = \sum_{i=1}^K \left (  w_i \frac{\partial \psi_i }{\partial \psi_k} + \psi_i \frac{\partial w_i }{\partial \psi_k} \right).
\end{align}
Similar to the previous subsection, notice that none of the weights $w_i$ are a function of any of the estimators $\psi_k$ (even when $i=k$). Thus, $\frac{\partial w_i}{\partial \psi_k}$ is always $0$. So the derivative simplifies to,
\begin{align}
\frac{\partial {\psi}}{\partial \psi_k} = \sum_{i=1}^K  w_i \frac{\partial \psi_i }{\partial \psi_k}.
\label{eq:partial-psi-1}
\end{align}
Finally, this simplifies nicely as,
\begin{equation}
    \frac{\partial \psi_i}{\partial \psi_k} =
    \begin{cases}
        1\  \text{ when } \hfill i=k \\
        0\ \text{ when } \hfill i\not=k.
    \end{cases}
\end{equation}
Plugging these in gives us the final expression for the partial derivative of $\psi$ with respect to $\psi_i$,
\begin{align}
\frac{\partial {\psi}}{\partial \psi_k} = w_k.
\label{eq:partial-psi-1}
\end{align}

\section{Proofs of Propositions~\ref{prop:backdoor-tests} and~\ref{prop:bdoor-fdoor-iv}}
\label{app:proposition-proofs}

\subsubsection*{Proof of Proposition~\ref{prop:backdoor-tests}}

\begin{proof}
    \cite{entner2013data} show under assumptions ${\cal A}_1$ and ${\cal A}_2$ in \eqref{eq:assumptions-backdoor} that $W \subseteq C$ is a valid backdoor adjustment if $Y\not\ci Z \mid W$ and $Y \ci Z \mid A, W$. The additional assumption ${\cal A}_3$ we make in \eqref{eq:assumptions-backdoor} ensures that $Y\not\ci Z \mid W$ is already true due to the existence of the path $Z \rightarrow A \rightarrow Y$. Thus, under assumptions ${\cal A}_1, {\cal A}_2, {\cal A}_3$, the independence $Y \ci Z \mid A, W$ alone is sufficient to ensure that model ${\cal M}_k$ is correct (i.e., $W$ is a valid backdoor adjustment set). The conclusion then follows, as we suppose that $\beta_k$ is an observed data parameter such that $\beta_k = 0 \iff Y \ci Z \mid A, W$.
\end{proof}

\subsubsection*{Proof of Proposition~\ref{prop:bdoor-fdoor-iv}}
\begin{proof}
The argument for $\beta_1 = 0$ implying correctness of the backdoor model with adjustment set $C$ is essentially the same as the proof of Proposition~\ref{prop:backdoor-tests} with $W=C$, since the assumptions used in Proposition~\ref{prop:backdoor-tests} are a superset of those in \eqref{eq:assumptions-backdoor}.

For the second implication, we make an argument similar to the one used in \cite{bhattacharya2022testability}. Under the Verma faithfulness assumption ${\cal A}_1$, we know that if $\beta_2 = 0$, the causal DAG $\G(V \cup U)$ must support identification of $P(A, Z, C, Y \mid \doo(m))$ by the g-formula and $Y \ci_{\text{d-sep}} Z \mid C$ in $\G_{\doo(m)}$. %Under the causal ordering assumptions ${\cal A}_2$, identification of this post-intervention distribution by the g-formula gives us $P(A, Z, C, Y \mid \doo(m)) = P(A, Z, C) \times P(Y\mid M=m, A, Z, C)$.
By assumption ${\cal A}_3$ we know that $Z \rightarrow A $ exists in $\G$. We now show that the existence of $A \rightarrow Y$ in $\G$ contradicts existence of the Verma constraint under faithfulness. Suppose $A \rightarrow Y$ does exist in $\G$. Then, $Y \not\ci_{\text{d-sep}} Z \mid C$ in $\G_{\doo(m)}$ due to the open path $Z\rightarrow A \rightarrow Y$, which is a contradiction. Thus, the frontdoor exclusion restriction of no $A \rightarrow Y$ is satisfied. Per \cite{tian2002general}, the  distribution $P(V\setminus \{A\} \mid \doo(a))$ where $A$ is a single treatment variable is identified if and only if there is no path from $A$ to any child $X$ of $A$ of the form $A \leftarrow \cdots \rightarrow X$ such that every collider on the path is an observed variable $V_i \in V$ and every non-collider on the path is an unmeasured variable in $U_i \in U$. By the causal ordering assumption ${\cal A}_2$ and the previous argument ruling out the existence of $A \rightarrow Y$, the only child of $A$ is $M$, so we only need to show that no such path exists from $A$ to $M$. The existence of such a path between $A$ and $M$ also contradicts the presence of the Verma constraint, as it would imply the  existence of some $U_i \in U$ that causes $M$ and thus also causes $Y$ (by assumption ${\cal A}_4$), preventing identification of $p(A, Z, C, Y \mid \doo(m))$. Thus, the second frontdoor restriction is satisfied and $P(Z, C, M, Y \mid \doo(a))$ is identified as \citep{tian2002general}
\begin{align*}
\left( \sum_{a'} P(a', Z, C) \cdot \E[Y\mid a', M, C, Z] \right) \cdot p(M \mid A=a, Z, C) \cdot
\end{align*}
It is straightforward then to obtain the frontdoor functional in \eqref{eq:frontdoor} for $\theta$ by summing over $Z, C$ to obtain $P(Y \mid \doo(a))$.

Finally for the third implication, $Z$ already satisfies the IV relevance condition by assumption ${\cal A}_3$. The second condition for IV validity is that $Z$ should be d-separated from $Y$ given $C$ in a graph where we delete the outgoing edges from $A$. This is satisfied when $M\ci Z \mid A, C$ and the Verma constraint holds per Corollary 1.1 in \cite{bhattacharya2022testability}.
\end{proof}

\section{Details Of Data Generating Processes}
\label{app:dgps}

In this subsection we describe the data generating processes (DGPs) underlying each of our numerical experiments. 
In our descriptions, we define $\text{expit}(x) \coloneqq 1/(1+\text{exp}(-x))$ for $x \in \mathbb{R}$. We use $\text{Bern}(p)$ as shorthand for the Bernoulli distribution with probability $p$ and $N(\mu, \sigma^2)$ as shorthand for the normal distribution with mean $\mu$ and variance $\sigma^2$.

\subsection{DGPs for Section~\ref{subsec:multi-bdoor}}
When the $U_1 \rightarrow Z$ edge is absent, the data are generated according to Figure~\ref{fig:app-dags}(a) as,
\begin{align*}
    Z &\sim \text{Bern}(0.5) \\
    U_1, U_2, C_1, C_2, C_3 &\sim N(0, 1) \\
    C_4 &\sim N\big(-2.5\cdot U_1 + 2\cdot U_2,\  1 \big) \\
    C_5 &\sim N\big(-2.5 \cdot U_1 + 2 \cdot U_2,\  1\big) \\
    A &\sim \text{Bern}\big(\text{expit}(2.75\cdot Z -3\cdot U_1 + C_1 + C_2 + C_3)\big) \\
    Y &\sim \text{Bern}\big(\text{expit}(1.5\cdot A + 2\cdot U2 + C_3 + C_4 + C_5)\big)
\end{align*}

When the $U_1 \rightarrow Z$ edge is present, the data are generated according to Figure~\ref{fig:app-dags}(a) as,
\begin{align*}
    U_1, U_2, C_1, C_2, C_3 &\sim N(0, 1) \\
    Z &\sim \text{Bern}\big(\text{expit}(U_1)\big) \\
    C_4 &\sim N\big(-2.5\cdot U_1 + 2\cdot U_2,\  1 \big) \\
    C_5 &\sim N\big(-2.5 \cdot U_1 + 2 \cdot U_2,\  1\big) \\
    A &\sim \text{Bern}\big(\text{expit}(2\cdot Z -3\cdot U_1 + C_1 + C_2 + C_3)\big) \\
    Y &\sim \text{Bern}\big(\text{expit}(A + 2\cdot U2 + C_3 + C_4 + C_5)\big)
\end{align*}

\subsection{DGPs for Section~\ref{subsec:bdoor-fdoor-iv}}
When both frontdoor and IV models are correct, data are generated according to Figure~\ref{fig:app-dags}(b) without the blue dashed edge as,
\begin{align*}
    Z &\sim \text{Bern}(0.5) \\
    C &\sim N(0, 1) \\
    U &\sim N(0, 1) \\
    A &\sim \text{Bern}\big(\text{expit}(2 \cdot Z + 2 \cdot C + 2 \cdot U)\big) \\
    M &\sim \text{Bern}\big( \text{expit}(-2 + 4\cdot A - 0.5\cdot C)\big) \\
    Y &\sim N\big(2\cdot M + 2 \cdot C + 2\cdot U, \ 1\big)
\end{align*}

When only the frontdoor model is correct, data are generated according to Figure~\ref{fig:app-dags}(b) with the blue dashed edge as,
\begin{align*}
    Z &\sim \text{Bern}(0.5) \\
    C &\sim N(0, 1) \\
    U &\sim N(0, 1) \\
    A &\sim \text{Bern}\big(\text{expit}(Z + C - 0.5 \cdot U)\big) \\
    M &\sim \text{Bern}\big( \text{expit}(-1 + 2\cdot A - Z + C)\big) \\
    Y &\sim N\big(2\cdot M - 0.75 \cdot C - 2\cdot U, \ 1\big)
\end{align*}

\section{Estimators Used In Numerical Experiments And Data Application}
\label{app:estimators}
Below we describe the specific estimators used in our numerical experiments and data application in more detail.

\subsection{Estimators Used In Section~\ref{subsec:multi-bdoor}}

\subsubsection*{Estimators for $\beta_k$}

For each log-odds ratio $\beta_k$, we use an influence function-based estimator of $\log(\text{OR}(Y, Z | A, W))$ from \cite{tchetgen2010doubly} and \cite{tan2019doubly}. An R implementation of their method is publicly available from Wu and Malinsky at \url{https://github.com/chaoqiw0324/ortest}; we translate this to Python for our purposes. First, let $\zeta(a, w) \coloneqq \E[Y \mid z_0, a, w]$ and $\eta(a, w) \coloneqq \E[Z \mid y_0, a, w]$. Let $O = \{Y, A, Z, W\}$. Then, an unbiased estimating function of $\beta_k$ is
\begin{align*}
    g(o, \zeta, \eta, \beta_k) = (y - \zeta(a, w))\cdot(z - \eta(a, w)) \cdot e^{-\beta_k\cdot y\cdot z}.
\end{align*}
Since $g$ is an estimating function, we can construct a point estimate $\beta_{k, n}$ of the log-odds ratio by first constructing estimators $\zeta_n$ and $\eta_n$ of $\zeta$ and $\eta$ respectively, and using these to obtain the value of $\beta_k$ that solves the equation $\sum_{i=1}^n g(o_i, \zeta_n, \eta_n, \beta_{k}) = 0$. This estimator is asymptotically linear under doubly robust conditions on the nuisance estimators $\zeta_n$ and $\eta_n$. Further, by standard Z-estimator theory, the influence function $\phi_{\beta_k} = B^{-1} g(o, \zeta, \eta, \beta_k)$, where $B^{-1} = \E[-\partial g/\partial \beta_k]$; see for example the review in \cite{cole2025five}. We can then construct an estimator of the influence function $\phi_{\beta_k, n}$ as
\begin{align*}
    \frac{1}{n}\sum_{i=1}^n \bigg( \frac{-\partial g(o_i, \zeta_n, \eta_n, \beta_k)}{\partial \beta_k} \bigg)^{-1} \cdot g(o_i, \zeta_n, \eta_n, \beta_{k}) \ \bigg\vert_{\beta_k = \beta_{k, n}}.
\end{align*}

\subsubsection*{Estimators for $\psi_k$}

For each ${\cal M}_k$ that uses an adjustment set $W$, we use the AIPW estimator \citep{bang2005doubly} of the backdoor formula \eqref{eq:backdoor} with $L = W$. Let $\pi(w) \coloneqq P(A=1 \mid W=w)$ and $\mu(a, w)\coloneqq \E[Y\mid A=a, W=w]$. Then the AIPW estimator is an asymptotically linear estimator of $\psi_k$ under doubly robust conditions on the nuisance estimators $\pi_n$ and $\mu_n$ of $\pi$ and $\mu$ respectively, with influence function $\phi_{\psi_k}$ given by
\begin{align*}
    \phi_{\psi_k} = \{ y - \mu(a, w)\} \bigg\{ \frac{a - \pi(w)}{\pi(w)(1-\pi(w))} \bigg\} + \{ \mu(1, w) - \mu(0, w) \} - \psi_k .
\end{align*}

\subsection{Estimators Used In Section~\ref{subsec:bdoor-fdoor-iv} and Data Application}

For these experiments we construct plug-in estimators of  $\beta_{k, n}$ and $\psi_{k, n}$ based on parametric nuisance estimation.

\subsubsection*{Estimators for $\beta_k$}
For $\beta_1$ and $\beta_3$, we construct parametric nuisance estimators $\zeta_n$ and $\eta_n$ of $\zeta(y, c, a, z) \coloneqq P(Y \mid A, C, Z)$ and $\eta(m, a, z, c) \coloneqq P(M\mid A, Z, C)$ respectively. In the parametric models we consider, estimates of the log-odds ratio $\log(\text{OR}(Y, Z \mid A, C))$ and $\log(\text{OR}(M, Z \mid A, C))$ are given by the coefficients of $Z$ in $\zeta_n$ and $\eta_n$ respectively.

For $\beta_2$, the parameter used to test the Verma constraint, let $\alpha\coloneq P(Y \mid Z, C)$ and $g(y, z, c, \alpha)$ be any unbiased estimating function used to construct a parametric nuisance estimator $\alpha_n$ of $\alpha$. The reweighted log-odds ratio $\log(\text{OR}_{\widetilde{P}}(Y, Z \mid C))$ is  obtained as the coefficient of $Z$,  one of the parameters estimated in $\alpha_n$, where $\alpha_n$ is the solution to a set of reweighted estimating equations $\sum_{i=1}^n g(y_i, z_i, c_i, \alpha)/\eta_n(m_i, a_i, z_i, c_i)=0$ and $\eta_n$ is a parametric nuisance estimator of $\eta$ as before. In terms of practical implementation, $\log(\text{OR}_{\widetilde{P}}(Y, Z \mid C))$ can be treated as the coefficient of $Z$ in a reweighted (linear or logistic) regression of the outcome on the covariates $C$ and anchor variable $Z$, where the weights are given by $1/\eta(m, a, z, c)$. 

\subsubsection*{Estimators for $\psi_k$}
For $\psi_1$, the backdoor functional \eqref{eq:backdoor} with $L=C$, we construct a plug-in estimator under parametric specification of the outcome regression model $\mu \coloneqq \E[Y\mid A, C]$. For $\psi_2$, the frontdoor functional \eqref{eq:frontdoor} with $L = C \cup \{Z\}$, we use a plug-in estimator of the more convenient dual IPW functional instead under parametric specification of the mediator model $\eta \coloneqq P(M\mid A, Z, C)$ \citep{fulcher2020robust, bhattacharya2022semiparametric}
\begin{align*}
\psi_{2, \text{Dual IPW}} = \E \bigg[ \frac{P(M\mid A=1, Z, C)}{P(M \mid A, Z, C)} \times Y \bigg] - \E \bigg[ \frac{P(M\mid A=0, Z, C)}{P(M \mid A, Z, C)} \times Y \bigg].
\end{align*}
Finally, for $\psi_3$, the IV functional \eqref{eq:iv} with $L=C$, we construct a plug-in estimator under parametric specification of the nuisance function in the numerator $\nu\coloneqq \E[Y\mid Z, C]$ and denominator $\xi \coloneqq \E[A \mid Z, C]$.

\end{document}